\documentclass[aps,prb,twocolumn,footinbib,superscriptaddress, 10pt]{revtex4-2}
\usepackage{amsmath, amssymb, amsthm}
\usepackage{algorithm}
\usepackage{algpseudocode}
\usepackage[normalem]{ulem}
\algblock{Input}{EndInput}
\algnotext{EndInput}
\algblock{Output}{EndOutput}
\algnotext{EndOutput}

\usepackage[colorlinks,bookmarks=false,citecolor=blue,linkcolor=blue,urlcolor=blue]{hyperref}
\usepackage{amsfonts}
\usepackage{amsmath}
\usepackage{cleveref}
\usepackage{bm}
\usepackage{xcolor}

\newtheorem{theorem}{Theorem}

\usepackage{booktabs}
\usepackage[most]{tcolorbox}

\usepackage{newtxtext,newtxmath}

\usepackage[utf8]{inputenc}
\usepackage{soul}
\usepackage{stackengine} 

\usepackage{changepage}

\usepackage{graphicx}
\graphicspath{ {./} }

\usepackage{orcidlink}

\newcommand{\bra}[1]{\left< #1 \right|}
\newcommand{\ket}[1]{\left| #1 \right>}

\newcommand{\tr}{\mathrm{Tr}}

\newcommand{\Xt}{\{X_t\}}

\newcommand{\Z}{\mathbb{Z}}
\newcommand{\C}{\mathbb{C}}
\newcommand{\E}{\mathbb{E}}

\begin{document}

\title{Maximum-Entropy Extension of Quantum Correlation Functions from Short Real-Time Dynamics}
\author{Filippo Maria Gambetta \orcidlink{0000-0002-9182-1415}}
\email{filippo@phasecraft.io}
\affiliation{Phasecraft Ltd., Colston Street, Bristol BS1 4XE, United Kingdom}
\author{Sabrina Yue Wang \orcidlink{0000-0001-8431-6683}}
\affiliation{Phasecraft Ltd., Colston Street, Bristol BS1 4XE, United Kingdom}
\author{Raul A. Santos \orcidlink{0000-0003-1161-1328}}
\affiliation{Phasecraft Ltd., Charlotte Street, London W1T 4PW, United Kingdom}

\date{\today}

\begin{abstract}
Correlation functions are central to understanding the behavior of quantum systems, with their spectra giving direct access to quasiparticle and collective excitations. However, the long-time evolution required for sufficient spectral resolution is challenging to reach for both classical and quantum simulations. In this work, we introduce a principled approach for extending time series of quantum correlation functions and reconstructing their spectral densities, based on their connection with covariance sequences of stationary stochastic processes. This connection enables us to apply Burg's maximum-entropy principle to find the least-committal extension compatible with the measured values of the correlation function, which corresponds to an all-pole autoregressive model whose coefficients solve a set of Yule--Walker equations. The resulting extension is positive semi-definite by construction, guaranteeing the positivity of the spectrum and the automatic satisfaction of its total spectral weight sum rule. The approach applies to both scalar and matrix-valued correlation functions, including those with highly structured or continuous spectra. The generalization to the matrix case makes it possible to trade simulation time for additional measured correlation functions, enabling accurate time-series extension and spectral reconstruction from short-time data, even in the presence of noise.
\end{abstract}

\maketitle

\section{Introduction}

Time-dependent correlation functions and their spectral densities are central to our understanding of many-body quantum systems, encoding how single-particle and collective excitations propagate and connecting directly to a wide range of experimentally accessible observables, such as quantum transport coefficients, optical and dynamical conductivities, dynamical structure factors, and spectral functions~\cite{BruusFlensberg2004, MartinReiningCeperley2016, GiulianiVignale2005}. In general, computing such observables is a challenging task for state-of-the-art numerical techniques, which are limited to small system sizes, short-time evolutions, and/or low-energy sectors~\cite{Schollwoeck_2011, Paeckel_2019, Kuhner_1999, Nocera_2016, Espinoza_2024}. On the other hand, while quantum computers hold the promise for efficient simulations of quantum dynamics~\cite{Lloyd_1996, Baez_2020}, current  NISQ devices are corrupted by noise. As a result, whether classical or quantum, time dynamics simulations of complex quantum systems are restricted to comparatively short times~\cite{Arute_2020, Alam_2025_google, Alam_2025_quantinuum, Vilchez_Estevez_2025, Granet_2026, Chowdhury_2026, Hartnett_2026, Ali_2025, Ting_2026, Lee_2026, Cochran_2024, Hogan_2026, Andersen_2026}. Motivated by this limitation, a variety of post-processing methods have been proposed to extend correlation functions (as well as other observables) from short-time dynamics to longer times or, equivalently, to enhance their resolution in the spectral domain beyond that directly accessible from the simulated time window. Many of these approaches exploit the fact that correlation functions of many-body quantum systems are, in general, a sum of complex exponentials, as follows directly from the Lehmann representation~\cite{BruusFlensberg2004}. Different methods exploit this structure in different ways. Linear prediction (LP)~\cite{Makhoul_1975_linear,White_2008_spectral,Barthel_2009_spectral,Ganahl_2015,Tian_2021} introduces an autoregressive model in which each time sample is expressed as a linear combination of previous ones. Prony, matrix-pencil, and subspace methods such as ESPRIT or MUSIC~\cite{StoicaMoses2005,Erpenbeck_2026,Gomes_2023,Kaneko_2025,Keller_2020,Roy_1989,Schmidt_1986,Stroeks_2022, Shee_2025} instead seek a finite set of exponential modes and directly infer their poles and amplitudes from the sampled signal. Dynamical mode decomposition (DMD)~\cite{Schmid_2010,Shen_2026,Yin_2021,Szasz_2025,Shen_2025} takes an intermediate approach: it learns an effective linear propagator between successive snapshots and extracts its dynamical modes and associated frequencies from its eigendecomposition.

When applied to short-time series, these methods perform well only for correlation functions with simple spectral densities with a few dominant peaks, but fail in the presence of structured or continuous spectra~\cite{Kemper_2026}. The latter are ubiquitous in condensed-matter systems, arising, for instance, from many-particle continua in the dynamical structure factors of interacting fermionic and spin models, or from coexisting quasiparticle resonances, incoherent continua, and spatially dependent features in the spectral functions of inhomogeneous systems with boundaries, disorder, or impurities~\cite{Nocera_2018, Nocera_2018b, Isidori_2010, Schoenauer_2018, Peters_2011, Ding_2024, Yan_2023}. The origin of this failure is a mismatch between model capacity and available data: from $n$ (complex) samples one can retrieve at most about $n/2$ frequencies~\cite{StoicaMoses2005}. So, when a spectrum is richer than the model can represent, these methods either retain only the $\sim n/2$ dominant frequencies or project the entire spectrum onto $\sim n/2$ of them, yielding poor reconstructions or spurious features. This, in general, cannot be fixed by sampling more finely, since only a longer evolution time adds resolving power to discriminate between nearby spectral lines. Moreover, these methods are prone to ill-conditioning, making them fragile to noise, which may arise, for instance, from finite bond dimension in tensor-network simulations or from shot noise in quantum simulations.

In this work, we discuss a method tailored to highly structured spectra for which ``sum-of-complex-exponential'' approaches usually fail. Instead of trying to resolve the modes making up the dynamics of a correlation function from a short time series, we focus on finding the ``best'' process that could have generated it compatibly with the known observed points. Here, ``best'' is rigorously defined in terms of Burg's maximum entropy principle for time series~\cite{Burg1975}. We first show that any time series of $n$ samples of a stationary correlation function can be identified with a segment of the \emph{covariance} sequence of an \emph{underlying} stationary stochastic process~\cite{Burg1975, GrenanderSzego1984}. In particular, when the Gram matrix associated with the correlation function is nonsingular (as is usually the case for short-time signals with a structured or continuum spectra) there are infinitely many admissible extensions of the observed time series. Among all of them, one natural choice is the one that commits to as little as possible about the covariance lags outside the observed window. In the language of stochastic processes, this means selecting the extension of maximal differential entropy rate. According to Burg's maximum-entropy theorem, this process is a Gaussian autoregressive process of order at most $n-1$, whose coefficients follow from a set of Yule--Walker equations built directly from the known data points~\cite{Burg1975, BrockwellDavis1991}. Moreover, its poles are all strictly outside the unit disk, ensuring a non-diverging extension, with no need for any ad hoc pole regularization. 
Crucially, the autoregressive form is here \emph{derived} from the maximum-entropy selection principle rather than postulated. The distinction from typical linear prediction is subtle but fundamental: the latter treats the correlator as a signal and \emph{estimates} covariance lags from it by time averaging, whereas here the correlator values \emph{are} the covariance lags. 
Since the samples of the correlation function are identified with the covariance lags of a stochastic process, its positive semi-definiteness and the total spectral weight sum rule are automatically enforced. The importance of preserving the positive semi-definiteness of correlation functions has been recently recognized and exploited for the their analytic continuation in imaginary time~\cite{Fei_2021_matrix, Fei_2021_nervalinna, Huang_2023, Dong_2026}, and extension in real-time~\cite{Kemper_2024, Kemper_2026}. We finally note that while maximum-entropy methods are routinely used for the analytic continuation of imaginary-time correlators obtained, for instance, from quantum Monte Carlo simulations~\cite{Silver_1990, Jarrell_1996, Levy_2017}, their application in real-time is much less developed. In the former case, the entropy is a Shannon-type functional of the spectral function relative to a default model, whereas in our work it is the (differential) entropy rate of the underlying stochastic process. Despite the different functionals, both implement the same principle: select the least-committal spectrum compatible with the observed data.

The maximum-entropy construction is not restricted to scalar data, and its extension to matrix-valued correlation functions follows naturally~\cite{Burg1975}. Identifying a sequence of $M$ matrix correlation functions with a finite segment of the covariance lags of a \emph{vector}-valued stationary process, the process that maximizes the entropy rate is a vector autoregressive (VAR) model whose coefficients are determined by a set of block Yule--Walker equations. As in the scalar case, the construction preserves by design the positive semi-definiteness of the matrix spectral density and the total spectral weight sum rule. Beyond enforcing the correct analytic structure of matrix-valued correlation functions, the vector formulation provides an additional advantage. By treating all correlators as components of a single vector process, it exploits their full joint covariance structure: auto- and cross-correlations are constrained simultaneously by the positive-semi-definite block-Toeplitz Gram matrix and, equivalently, by the positivity of the matrix spectral density. This coupling is physically motivated, since all observables evolve under the same Hamiltonian and share the same underlying set of excitation energies. Consequently, information carried by one correlator constrains the admissible continuation of the others. We show that this joint, covariance-consistent modeling can substantially improve both predictive accuracy and robustness to noise compared with treating the correlators independently.

To showcase the capabilities of our approach, we apply it to the retarded Green's function of a one-dimensional Fermi-Hubbard chain with open boundary conditions. We show that, although both scalar and vector approaches are capable of producing extensions of the retarded Green's function in good agreement with the exact one, the vector model substantially outperforms its scalar counterpart: it faithfully extends the original correlator time series by nearly an order of magnitude, and reconstructs the local density of states (i.e., the site-dependent spectral function) to high accuracy even from short and noisy time series.

\section{meAR($p$) modeling for correlators}
\label{sec:arp}

When performing a classical or quantum simulation of the dynamics of some correlation function $g(t)$, one usually has access to the values of $g(t)$ on an equally-spaced time grid, i.e., $g_k := g(t_k)$ with $t_k = k\Delta t$ for $k=0,\dots,n-1$, that is to the time series $\{g_k\}_{k=0}^{n-1}$.
For the sake of simplicity, here we focus on a zero-temperature stationary correlator $g(t) = \bra{\psi} \mathcal{O}^\dagger(t)\mathcal{O}\ket{\psi}$ for some operator $\mathcal{O}$, with $\ket{\psi}$ a stationary state. All of what we say can be generalized to finite-temperature correlators straightforwardly. 
The Lehmann representation of $g(t)$ is given by~\cite{BruusFlensberg2004, GiulianiVignale2005}
\begin{equation}
	\label{eq:gt_lehmann}
	g(t) = \int_{-\infty}^{\infty} S(\omega) e^{-i \omega t} d\omega,
\end{equation}
with non-negative spectral density $S(\omega)=\sum_m W_m \delta(\omega-\omega_m) \geq 0$, with $W_m = |\bra{m} \mathcal{O} \ket{\psi}|^2 \geq 0$ and $\omega_m = E_m - E_0$. The total spectral weight is fixed by the sum rule $g(0) = \int S(\omega) d\omega$. \Cref{eq:gt_lehmann} implies that $g(t)$ is a positive semi-definite (PSD) function \cite{Kemper_2024}, namely that for any finite sequence of times $t_0,\dots,t_{n-1}$, one has
$ \sum_{ij}\lambda_i^*\lambda_j g(t_i-t_j)\geq0$, $\forall \lambda_i\in\mathbb{C}$.
In turn, this implies that the Hermitian Toeplitz matrix (also known as Gram matrix \cite{Kemper_2024})
\begin{equation}
	\label{eq:gt_Gram}
	R^{(n)} = \begin{bmatrix}
		g_0 & g_1^*  & \cdots & g_{n-1}^*\\
		g_1 & g_0 & \cdots & g_{n-2}^*\\
		\vdots & \vdots  & \ddots  & \vdots\\
		g_{n-1} & g_{n-2} & \cdots & g_0\\
	\end{bmatrix},
\end{equation}
where we defined $g_k:=g(t_k)$ and used $g_{-k}=g_k^*$ due to stationarity, is PSD, i.e.\ $R^{(n)}\succeq 0$. 

The spectral density $S(\omega)$ introduced in \cref{eq:gt_lehmann} contains a wealth of information about many-body quantum systems. For instance, the spectral density of the Green's function (also known as ``the'' spectral function) probes the behavior of quasiparticle excitations and enters into many calculations of experimentally relevant properties, such as ARPES and transport coefficients, while the spectra of response functions (often called susceptibilities) offer a direct window onto the behavior of collective excitations and can be probed, for instance, by inelastic neutron or X-ray scattering experiments~\cite{Sturm_1993, Damascelli_2003,Ament_2011_resonant, Zaliznyak_2015, MartinReiningCeperley2016}. Given their importance, the goal of many classical and quantum simulations is to compute the time series $\{g_k\}_{k=0}^{n-1}$ of some correlation function and then its spectral density \footnote{We note that there exist methods to directly compute the spectral functions~\cite{Kunher_1999, Jeckelmann_2002, Nocera_2016}.}. The simplest way to do so is via a (discrete) Fourier transform: $S(\omega_m) = (2\pi)^{-1}\sum_{k=0}^{n-1} g_k e^{i k \omega_m}$, with $\omega_m = 2\pi m/n$ and $m = 0,\dots,n-1$. In this case, the spectral resolution is limited by the number of samples as $\Delta \omega = 2\pi/n\Delta t$. In practice, this choice corresponds to making a strong assumption on the value of the correlator time series outside the known interval, i.e., $\{g_k\}_{k<0} = 0$ and $\{g_k\}_{k\geq n} = 0$. 
On the other hand, methods such as linear prediction (LP)~\cite{Makhoul_1975_linear,White_2008_spectral,Barthel_2009_spectral,Ganahl_2015,Tian_2021} and subspace methods~\cite{StoicaMoses2005,Erpenbeck_2026,Gomes_2023,Kaneko_2025,Keller_2020,Roy_1989,Schmidt_1986,Stroeks_2022} assume for the correlator a strictly deterministic form, $g(t) = \sum_{m = 1}^{P} W_m e^{-i \omega_m t}$ for some value of $P$. While this functional form is physically justified by \cref{eq:gt_lehmann}, one does not know the number of complex frequencies $P$ a priori, which has to be treated as a hyperparameter. Moreover, as explained in the introduction, these methods can only resolve $\sim n/2$ frequencies and yield poor results when applied to correlation functions with a highly structured spectrum. 

But what if, instead, one looks for the spectrum obtained from assuming as little as possible about the future points $\{g_k\}_{k\geq n}$, with the only constraint of being a PSD sequence compatible with the known observations~\cite{CoverThomas2006}? This is exactly the framework in which Burg developed his maximum entropy approach to compute power spectral densities from a finite sequence of covariance lags~\cite{Burg1975}. One complication he had to face was to obtain reliable covariance lags from the finite samples of a given time series, which led him to design the well-know Burg's algorithm for reflection coefficients~\cite{Burg1975}. We now show that, when dealing with correlation functions of quantum systems, we are in a much better position since we have direct access to the covariance lags we need to compute Burg's maximum entropy spectrum. To see this, let us go back to the samples $\{g_k\}_{k=0}^{n-1}$ of a correlation function $g(t)$. Since $g(t)$ is a PSD function, $\{g_k\}_{k=0}^{n-1}$ is a segment of a PSD sequence. Then, Burg's fundamental autocorrelation matrix theorem~\cite{Burg1975} (which, in this context, is equivalent to the Carath\'{e}odory-Toeplitz extension theorem~\cite{GrenanderSzego1984}) ensures that there exists at least one PSD extension of $\{g_k\}_{k=0}^{n-1}$ to the whole $\mathbb{Z}$, $\{g_k\}_\mathbb{Z}$. In particular, if the Toeplitz Gram matrix associated with $\{g_k\}_{k=0}^{n-1}$ is strictly positive definite, $R^{(n)} \succ 0$, then there exist infinitely many such PSD extensions. On the other hand, if $R^{(n)}$ is singular, $\det R^{(n)} = 0$, i.e., if its rank is $P < n$, then the PSD extension is unique. In the latter case, we can write $g(t)$ as a sum of $P$ complex exponentials, $g(t) = \sum_{m = 1}^{P} W_m e^{-i \omega_m t}$. This is exactly the assumption underlying all LP and subspace signal processing methods. In practice, for correlation functions with structured spectra and/or short time series, the number of frequencies $P$ contributing to the dynamics of a correlation function can be $P\gg n$ and/or their spectral gaps $\Delta E$ can be too small to be resolved. In this case $R^{(n)} \succ 0$ and, as we said, infinitely many PSD extensions are compatible with the observed time series $\{g_k\}_{k=0}^{n-1}$~\footnote{Note that, in principle, what follows can also be applied to the singular or near-singular case via a ridge regularization, $R^{(n)}\leftarrow R^{(n)}+\lambda I$. This may be appropriate to obtain well-conditioned extensions for rank-deficient noisy Gram matrices. Or, as explained in \cref{app:sec:burg}, one can apply it to a reduced Gram matrix $R^{(n')}$ with $n' = \max\{k\leq n: R^{(k)}\succ 0\}$.}. In such cases, standard LP or subspace methods would allow one to retrieve only a (potentially distorted) subset of the spectral frequencies. Importantly, any of these infinitely many PSD extensions can be realized as the covariance sequence of a wide-sense stationary (WSS) discrete stochastic process $\{X_t\}_{t\in\mathbb{Z}}$, i.e., a process whose first and second moments are time-invariant~\cite{BrockwellDavis1991}; See \cref{app:sec:stochastic_processes}. We recall that, given a zero-mean discrete WSS stochastic process $\{X_t\}_{t\in\mathbb{Z}}$, its covariance lags are defined as
\begin{equation}
	\label{eq:autocov_lags}
	\gamma_k := \mathbb{E}[X_{t+k}X^*_{t}].
\end{equation}
The sequence of covariance lags $\{\gamma_k\}_{k\in\mathbb{Z}}$ is PSD and, by virtue of the Bochner-Herglotz theorem, it admits a representation in terms of a non-negative (power) spectral density $f(\theta)\geq 0$ via~\cite{BrockwellDavis1991}
\begin{equation}
	\label{eq:gamma_PSD}
	\gamma_k = \frac{1}{2\pi}\int_{-\pi}^{\pi}f(\theta)e^{i\theta k} d\theta.
\end{equation}
Assuming the spectral density $S(\omega)$ of the correlator $g(t)$ introduced in \cref{eq:gt_lehmann} has finite bandwidth $W$, i.e. $S(\omega)\approx0$ for $|\omega-\omega_0| > W$ for some $\omega_0$, which is a realistic assumption for typical quantum many-body systems, and that the sampling interval $\Delta t$ is such that $W\leq \pi/\Delta t$ (Nyquist criterion), we can link $f(\theta)$ to the spectral density of $g(t)$ via
\begin{equation}
    S(\omega) = \frac{\Delta t}{2\pi} f(-\omega \Delta t).
\end{equation}

So, the correlator time series $\{g_k\}_{k=0}^{n-1}$ can be interpreted as a segment of a sequence of the covariance lags of infinitely many WSS stochastic processes. Among these, one has particular importance: the process with covariance lags matching $\{g_k\}_{k=0}^{n-1}$ that maximizes the differential entropy rate (see \cref{app:sec:entropy_rate})~\cite{Burg1975, Choi_1987, CoverThomas2006}. Intuitively, maximizing the differential entropy rate corresponds to assume maximum uncertainty (i.e., being least committal) about the unknown lags $\{g_k\}_{k\geq n}$. Burg's maximum entropy 
rate theorem (see \cref{app:sec:burg}) states that this process is a proper complex Gaussian autoregressive (AR) process $X_t$ of order (at most) $p = n-1$, AR($p$), of the form $X_t = \sum_{k=1}^{n-1}a_k X_{t-k} + \varepsilon_t$ with coefficients $a_k\in\mathbb{C}$ obtained from the solution of the Yule--Walker (YW) equations~\cite{BrockwellDavis1991, Burg1975, Neese_1993, CoverThomas2006, SchreierScharf2010}
\begin{equation}
	\label{eq:AR_cov_lag_prediction}
	g_m = \sum_{k=1}^{n-1}a_k g_{m-k}, \qquad m = 1,\dots,p,
\end{equation}
and innovation $\varepsilon_t \sim \mathcal{CN}(0, \sigma_{\varepsilon}^2)$, sampled from a complex normal distribution with zero mean and variance $\sigma_\varepsilon^2$, with $\mathbb{E}[\varepsilon_t X^*_s] = 0$ for  $s<t$ and  innovation variance $\sigma^2_{\varepsilon} = g_0 - \sum_{k=1}^{n-1}a_k g^*_k$. We recall that a proper complex Gaussian process is a complex-valued Gaussian process with vanishing pseudo-covariance, $\mathbb{E}[X_{t}X_s] =0,\ \forall t,s$~\cite{Neese_1993,SchreierScharf2010}. See \cref{app:sec:stochastic_processes} for further details. Note that the Gaussianity of the process is not an assumption about the physical system: it is a property of the maximum entropy rate process at fixed covariance lags~\cite{Choi_1987}. Once the coefficients of the AR($p$) process are known, the original correlator time series can be recursively extended via 
\begin{equation}
    \hat{g}_{m} = \sum_{k=1}^{n-1}a_k g_{m-k}\quad\text{for } m\geq n.
\end{equation}
Note that this extension is guaranteed to be PSD. Moreover, the maximum entropy spectrum of $g(t)$ compatible with the observed samples $\{g_k\}_{k=0}^{n-1}$ can be directly obtained from the AR($p$) coefficients as
\begin{equation}
	\label{eq:maxent_spectrum}
	S_{\rm ME}(\omega) = \frac{\Delta t}{2\pi}\frac{\sigma^2_{n-1}}{\left|1-\sum_{k=1}^{n-1}a_k e^{-ik\omega \Delta t}\right|^2},
\end{equation}
without the need of any time extension and Fourier transform~\cite{Burg1975}. See \cref{app:sec:stochastic_processes,app:sec:spectral_density}. Importantly, $S_{\rm ME}(\omega)$ is guaranteed to be non-negative and satisfies the total spectral weight sum rule, $\int S_{\rm ME}(\omega) d\omega = g_0$, by construction. In practice, computing $S_{\rm ME}(\omega)$ from \cref{eq:maxent_spectrum} requires a careful choice of the grid of $\theta$ values and may become ill-conditioned if $R^{(p)}$ is near singular~\cite{Burg1975}. A more robust approach is to introduce a small ridge regularization parameter $\lambda$, $R^{(p)} \leftarrow R^{(p)} + \lambda I$, solve the YW equations, extend the correlation function to the desired time, and take a conventional discrete Fourier transform. Here, $\lambda$ is chosen as the smallest possible value such that the whole extension is PSD up to the desired time. 

\Cref{eq:maxent_spectrum} highlights one of the main differences between our approach and standard LP and subspace methods. The maximum entropy spectrum is characterized by $p$ poles, as follows from the degree-$p$ denominator in the complex variable $z=e^{-i\omega\Delta t}$ in \cref{eq:maxent_spectrum}. In contrast with LP and subspace methods, however, these poles do not need to correspond to individual excitations of the system. Rather, they are chosen so that the resulting rational spectral density is the maximum-entropy spectrum compatible with the observed covariance lags. This distinction is particularly important for continuous spectra. Whereas LP and subspace methods approximate the signal by a finite set of discrete exponential modes, the AR maximum-entropy construction uses its poles collectively to provide a rational approximation to the full spectral density. Thus, even when no individual AR pole has a direct physical interpretation, their combined contribution can efficiently represent broad continua and other non-discrete spectral features. Moreover, our approach naturally solves the following three issues one has to deal with when using standard LP or subspace methods: 
\begin{enumerate}
    \item There is no need to discard the initial samples of $\{g_k\}_{k=0}^{n-1}$ to get rid of transient effects: the whole time series $\{g_k\}_{k=0}^{n-1}$ represents the beginning of the covariance lags sequence of some underlying WSS process;
    \item there is no need to treat the AR process order $p$ (or, equivalently, the number of frequencies making up the signal) as a hyperparameter: for a non-singular Gram matrix $R^{(n)}$, the maximum entropy rate solution fixes the order as $p = n-1$;
    \item PSD of the correlator extension, non-negativity of the spectral density, and total spectral weight sum rule are ensured by construction.
\end{enumerate} 
We discuss some further important differences between our approach and standard LP in \cref{app:sec:linear_prediction}.

To conclude this section, we summarize here the general prescriptions of the method we introduced above, which we dub maximum entropy AR (meAR):
\begin{tcolorbox}[colback=blue!5!white, colframe=blue!75!black]
Given $n$ samples of a correlation function $g(t)$, $\{g_k\}_{k=0}^{n-1}$, build the corresponding Gram matrix $R^{(n)}$. If $R^{(n)}\succ 0$, identify $\{g_k\}_{k=0}^{n-1}$ with covariance lags of some (unknown) WSS stochastic process, solve the associated YW equations, and extend the sequence; equivalently, compute its spectrum via the associated maximum entropy AR($n-1$) process. 
\end{tcolorbox}

\section{meVAR($p$) models for matrix correlators}
\label{sec:varp}

In this section we show that the meAR approach discussed in \cref{sec:arp} can be generalized to matrix correlators $G_{ij}(t) = \bra{\psi} \mathcal{O}^\dagger_i(t)\mathcal{O}_j\ket{\psi}$, with $i,j = 1,\dots,M$. Here, $i,j$ are arbitrary labels, e.g., spatial sites, spin or orbital indices, or a combination of them. Such correlators are ubiquitous in many-body physics. For instance, the case with $i,j$ as spatial site indices enters the calculation of momentum-resolved spectral and response functions. While meAR can be applied individually to all the diagonal correlators $G_{ii}(t)$, off-diagonal correlators with $i\neq j$ are not PSD and, therefore, cannot be identified with covariance lags. One option, suggested in Ref.~\cite{Kemper_2024}, is to apply meAR pairwise: for each pair $i,j$, introduce the linear combinations $B = \mathcal{O}_i + \mathcal{O}_j$ and $C = \mathcal{O}_i + i \mathcal{O}_j$. Then, the off-diagonal correlators can be written as linear combinations of $\langle \mathcal{O}^\dagger_i(t)\mathcal{O}_i\rangle$, $\langle \mathcal{O}^\dagger_j(t)\mathcal{O}_j\rangle$, $\langle B^\dagger(t)B\rangle$, and $\langle C^\dagger(t)C\rangle$, all of which are PSD and can then be extended via meAR. As we will show below, the matrix correlator $G_{ij}(t)$ is a matrix-valued PSD kernel and this approach does not ensure its overall PSDness. Hence, we expect this approach to work only for matrix correlators with small off-diagonal components.

One advantage of the maximum entropy approach we introduced in the previous section is that it can be generalized to matrix correlators in a straightforward way. First, to see that $G$ is a matrix-valued PSD kernel, we define $X := \sum_{k=0}^{n-1}\sum_{i=1}^{M}\lambda_{k,i}\mathcal{O}_i(t_k)$. Since $\langle X^\dagger X\rangle \geq 0 $, $ \forall \bm{\lambda}\in\mathbb{C}^{nM}$ it follows that
\begin{equation}
    \sum_{k,m=0}^{n-1}\sum_{i,j=1}^{M}\lambda_{k,i}^*\lambda_{m,j}G_{ij}(t_k-t_m)  \geq 0.
\end{equation}
Similarly to the scalar case, defining the matrices $\mathbb{C}^{M\times M} \ni G_{k-m} := G(t_k-t_m)$ and using the stationarity condition $G_{-k} = G_k^\dagger$, we can build the block Gram matrix of $G_{ij}(t)$,
\begin{equation}
	\label{eq:Gt_block_gram}
	\mathcal{R}^{(n)} = 
	\begin{bmatrix}
		G_0 & G^\dagger_1  & \cdots & G^\dagger_{n-1}\\
		G_1 & G_0  & \cdots & G^\dagger_{n-2}\\
		\vdots  & \vdots & \ddots  & \vdots\\
		G_{n-1} & G_{n-2}  & \cdots & G_0\\
	\end{bmatrix}\succeq 0,
\end{equation}
which is a Hermitian PSD block-Toeplitz matrix. Note that here each entry (block) is a matrix $G_k\in\mathbb{C}^{M \times M}$, and the Toeplitz structure holds only at block level and not for the individual components. The Lehmann representation of the matrix correlator $G_{ij}(t)$ is 
\begin{equation}
	G_{ij}(t) = \int_{-\infty}^{\infty} S_{ij}(\omega)e^{-i\omega t} d\omega,
\end{equation}
with $S(\omega)\in\mathbb{C}^{M \times M}$ the PSD spectral density matrix with entries
\begin{equation}
	\label{eq:matrix_lehmann}
	S_{ij}(\omega) = \sum_{m} \bra{\psi}\mathcal{O}_i^\dagger\ket{m}\bra{m} \mathcal{O}_j \ket{\psi}  \delta(\omega-\omega_m).
\end{equation}
In this case, we have a matrix total spectral weight sum rule, $G_{ij}(0) = \int S_{ij}(\omega)d\omega$. Note that in \cref{eq:matrix_lehmann} different indices $i,j$ enter the spectral density matrix with the same factor $\delta(\omega-E_m)$. Therefore, multiple correlators share the same poles, although individual correlators may couple only to subsets of them and with different weights. This is a consequence of the fact that their dynamics is generated by the same Hamiltonian. 

Everything we discussed for the scalar case can be applied almost directly to the matrix case. In particular, given a matrix time series $\{G_{ij, k}\}_{k=0}^{n-1}$, with $G_{ij, k} := G_{ij}(t_k)$ with $t_k = k\Delta t$, $k=0,\dots,n-1$, the fundamental multichannel autocorrelation theorem~\cite{Burg1975} guarantees that if $\mathcal{R}^{(n)}\succ0$ there exist infinitely many PSD extensions of $\{G_{ij, k}\}_{k=0}^{n-1}$ to the whole $\mathbb{Z}$. Each of these sequences is the covariance sequence of infinitely many $M$-dimensional WSS vector-valued stochastic processes $\bm{X}_t = (X_{1,t},..,X_{M,t})^T$~\cite{BrockwellDavis1991} and, by virtue of the matrix Bochner-Herglotz theorem, it admits a representation in terms of a matrix-valued spectral density $F(\theta)$. Among all of them, the one that maximizes the differential entropy rate is given by a proper complex Gaussian vector autoregressive (VAR) model~\cite{Burg1975, Neese_1993, SchreierScharf2010, Lutkepohl2005} of order at most $p=n-1$ with coefficient matrices $A_1,\dots, A_p \in\mathbb{C}^{M \times M}$ determined by the solution of the multichannel YW equations,
\begin{equation}
    \label{eq:matrix_yw}
	G_m = \sum_{k=1}^{n-1}A_kG_{m-k} ,\qquad m = 1,\dots,n-1;
\end{equation}
See \cref{app:sec:burg}. Once \cref{eq:matrix_yw} is solved for the coefficient matrices $\{A_k\}_{k=1}^{n-1}$, the correlator time series can be extended via 
\begin{equation}
    \hat{G}_{m}= \sum_{k=1}^{n-1}A_kG_{m-k}\quad\text{for } m\geq n.
\end{equation}
Let us recall that a proper complex Gaussian VAR($p$) process is written as
$\bm{X}_t = \sum_{k=1}^{p}A_k\bm{X}_{t-k} + \bm{\varepsilon}_t
$, with innovation vector $\bm{\varepsilon}_t\sim\mathcal{CN}(0,\Sigma_{\bm{\varepsilon}})$ such that $\mathbb{E}[\varepsilon_t \bm{X}^\dagger_s] = 0$ and $\mathbb{E}[\bm{\varepsilon_t} \bm{X}^T_s] = 0$ for $s<t$~\cite{SchreierScharf2010}. The innovation covariance matrix is defined as $\mathbb{E}[\bm{\varepsilon}_t \bm{\varepsilon}n_s^\dagger] = \Sigma_{\bm{\varepsilon}}\delta_{ts}$, with $\Sigma_{\bm{\varepsilon}} = G_0 - \sum_{k=1}^{p}A_k G^\dagger_k$. For a proper Gaussian vector process we have $\mathbb{E}[\bm{\varepsilon}_t \bm{\varepsilon}_s^T] = 0$. See \cref{app:sec:var} for further details.

Analogously to scalar case, the density spectral matrix of the maximum entropy VAR($p$) process can be obtained directly from its coefficients~\cite{Burg1975, LindquistPicci2015}. Assuming $S_{ij}(\omega)$ has compact support, i.e., $S_{ij}(\omega) \approx 0$ $\forall i,j$ for $|\omega-\omega_0| > W$ for some $\omega_0$ and $W\leq \pi/\Delta t$, the maximum entropy matrix spectral density of $G_{ij}(t)$ can be obtained as 
\begin{equation}
    S_{\rm ME}(\omega) = \frac{\Delta t}{2\pi} [A(e^{-i\omega \Delta t})]^{-1}\Sigma_{\bm{\varepsilon}} [A(e^{-i\omega \Delta t})^\dagger]^{-1},
\end{equation}
with $A(z) = I-\sum_{k=1}^{n-1}A_kz^k$ the matrix-valued AR polynomial. See \cref{app:sec:var}.

So, also in the matrix correlator case, the prescription of our maximum entropy VAR (meVAR) method is as follows: 

\begin{tcolorbox}[colback=blue!5!white, colframe=blue!75!black]
	Given the time series of a matrix correlator $\{G_{ij,k}\}_{k=0}^{n-1}$, build the corresponding block Gram matrix $\mathcal{R}^{(n)}$. If $\mathcal{R}^{(n)}\succ$0, interpret $\{G_{ij,k}\}_{k=0}^{n-1}$ as the beginning of the covariance sequence of a vector WSS stochastic process, solve the associated multichannel YW equations, and extend the matrix time series; equivalently, compute its spectral density matrix via the associated maximum entropy VAR($n-1$) process.
\end{tcolorbox}

As noted below \cref{eq:matrix_lehmann}, the components of the spectral density matrix of a many-body system share the same underlying set of excitation energies determined by the Hamiltonian. Intuitively, jointly modeling several correlations function should therefore provide more information about this common spectral structure than any single correlator over the same training-time window. A VAR($p$) model is capable of capturing both temporal correlations within individual components and cross-correlations between them, allowing it to efficiently exploit dynamical information shared across multiple correlators. 

Recent approaches have begun to exploit information beyond a single scalar time series, including ``space--time'' linear prediction~\cite{Tang_2025} and multi-observable DMD~\cite{Shen_2026}. Related multichannel ideas have a longer history in cross-correlated harmonic inversion~\cite{Main_1999} and, more recently, in matrix-valued ESPRIT~\cite{Zhang_2024,Dong_2026}. These approaches couple different signals through a common predictive model or shared spectral representation, and typically require choices such as prediction ranges, model rank, singular-value truncation, pole stabilization, or other regularization parameters. Moreover, physical constraints such as positivity and spectral normalization are not, in general, consequences of the underlying multichannel fitting procedure itself, although they can be incorporated explicitly in some formulations~\cite{Zhang_2024,Dong_2026}. In contrast, meVAR exploits the additional structure that arises when the signals form a matrix-valued covariance sequence: auto- and cross-correlations are treated jointly and constrained through their PSD block-Toeplitz Gram matrix. Rather than coupling different channels solely through a common dynamical model or shared spectral poles, meVAR therefore enforces finite-lag covariance consistency by construction. This ensures the PSDness of the extension of the matrix correlator and of its matrix spectral density, and preservation of the total spectral-weight sum rule. Our approach requires comparatively little parameter tuning and remains computationally inexpensive up to a large number of correlators since it requires the manipulation of $nM \times nM$ matrices. As we show below, this joint modeling of multiple correlators leads to improved extrapolation accuracy, spectral resolution, and robustness to noise. This makes meVAR particularly attractive not only when matrix-valued correlation functions are required, but also for extrapolating individual scalar correlators when the additional cross-correlators can be measured at an acceptable cost.

\section{Example: 1D Fermi-Hubbard chain with open boundaries}
\label{sec:1dhubbard}
In this section we apply the maximum entropy approach for the extension of correlation functions we introduced in this work to a 1D Fermi-Hubbard model with length $L$ and open boundary conditions with Hamiltonian
\begin{align}
	H &= -t_h\sum_{i=0}^{L-2}\sum_{\sigma}(c^{\dagger}_{i,\sigma} c_{i+1,\sigma} + c^{\dagger}_{i+1,\sigma} c_{i,\sigma}) \nonumber \\ &+U\sum_{i=0}^{L-1}n_{i,\uparrow}n_{i,\downarrow} - \mu \sum_{i=0}^{L-1} n_i.
\end{align}
Here, $c^{\dagger}_{i,\sigma}$ and $c_{i,\sigma}$ are fermionic creation and annihilation operators, $n_{i,\sigma} = c^{\dagger}_{i,\sigma} c_{i,\sigma}$ is the spin-resolved local density, and $n_i = n_{i, \uparrow} + n_{i, \downarrow}$ is the total local density. We apply the meAR and meVAR approaches to the spin-summed retarded Green's function $G^{R}_{ij}(t) = -i \theta(t) C^R_{ij}(t) $, with
\begin{equation}
	\label{eq:CR}
	C^{R}_{ij}(t) := \sum_\sigma\bra{\psi}\{c_{i,\sigma}(t), c^\dagger_{j, \sigma}\}\ket{\psi},
\end{equation}
$\theta(t)$ Heaviside function, and $\{A, B\} = AB + BA$ the usual anticommutator. Note that $C^{R}_{ij}(t)$ is the sum of two matrix-valued PSD kernels and is therefore PSD (while $G^{R}_{ij}(t)$ is not). Hence, me(V)AR can be applied either to $C^{R}_{ij}(t)$ or to its two components individually. In general, the latter approach gives better results in the noiseless case, but it is less robust in the presence of noise. In this section, we choose the former approach. Due to the presence of open boundary conditions, $C^{R}_{ij}(t)$ is not translationally invariant and cannot be reduced to a scalar correlator via Fourier transform. Therefore, to compute $C^{R}_{ij}(t)$ using meAR we need to make use of the pairwise strategy discussed at the beginning of \cref{sec:varp}~\cite{Kemper_2024}. In what follows, we consider a 1D chain with $L=11$ and set $t_h=1$, $U=4$, $\mu = 2$, and the number of fermions in each spin sector as $N_{\uparrow} = N_{\downarrow} = 5$. In what follows, times will be measured in units of $t_h$. For both meAR and meVAR, we will use the correlator time series sampled at times $t_k=k\Delta t $, $k = 0,\dots, n-1$. Hence, the order $p$ of both meAR and meVAR is set to $p = n-1$. 

In \cref{fig:CR_spacetime} we compare the exact space-time dynamics of $C^R_{i5}(t)$ (left) with its extensions from $t=3$ obtained via pairwise meAR (middle) and meVAR (right). Note that this time series is extremely short and, at large separation $|i-j|$, $C^{R}_{ij}(t)\approx0$ throughout much of the training window. It is therefore unsurprising that meAR can only faithfully extend $C^{R}_{ij}(t)$ for a limited time before the predicted signal gradually decays towards zero. In contrast, the meVAR extension, constructed from the full $11\times11$ matrix of correlators $C^{R}_{ij}(t)$, remains in excellent agreement with the exact dynamics for substantially longer times. In particular, it is able to capture all the nontrivial space-time features such as the propagating light cone and its subsequent reflections from the boundaries. This improvement is a direct manifestation of the joint multichannel modeling: rather than extending weak or nearly vanishing correlators from their own short history, meVAR exploits the auto- and cross-correlations among all matrix elements, allowing dynamical information contained in strongly resolved channels to constrain and stabilize the extrapolation of more weakly resolved ones. In this way, the shared spatiotemporal structure of the correlation matrix provides substantially more predictive information than is available to any individual scalar time series.

In \cref{fig:LDoS} we show that this improvement is even more striking when looking at the local density of states (LDoS),
\begin{equation}
	A_{ii}(\omega) := -\frac{1}{\pi}\mathrm{Im}G^{R}_{ii}(\omega + i\eta),
\end{equation}
with $\eta > 0$ a Lorentzian broadening factor, obtained by taking the Fourier transform of the extended correlators (or, equivalently, computed from the maximum entropy formula for the spectral densities). The first thing we note is that, for each $i$ the exact LDoS (left panel) is highly non trivial, containing contributions from several frequencies. This is exactly the case with which methods based on the ``sum-of-complex-exponentials'' approaches would struggle: given the short time series, they would only be able to identify a few of these frequencies. In contrast, from the central panel of \cref{fig:LDoS} we can see that, even for such a short time series, the LDoS obtained via meAR is a good approximation of the exact one: instead of a few peaks, we have a smoothened version of the exact LDoS. Similarly to \cref{fig:CR_spacetime}, the quality of the spectral reconstruction improves dramatically when multiple correlators are modeled as a meVAR process: the LDoS obtained via meVAR (right panel) is in perfect agreement with the exact one, with all its non-trivial features captured faithfully.

In \cref{fig:CR_offdiag_vs_radii} we investigate how the performance of meVAR depends on the number of correlators included in the model. In particular, we compare the extensions of the off-diagonal correlator $C^{R}_{45}(t)$ from $t = 3.0$ using pairwise meAR and meVAR as a function of the number of neighboring sites included in the VAR model. While the two approaches behave similarly when only the original sites $4$ and $5$ are included, the accuracy of the meVAR extension improves significantly as more sites/correlators are added, achieving an almost perfect agreement up to $t\simeq 10$ when all $11$ sites of the chain are included. 

\begin{figure*}
	\centering
	\includegraphics[width=\textwidth]{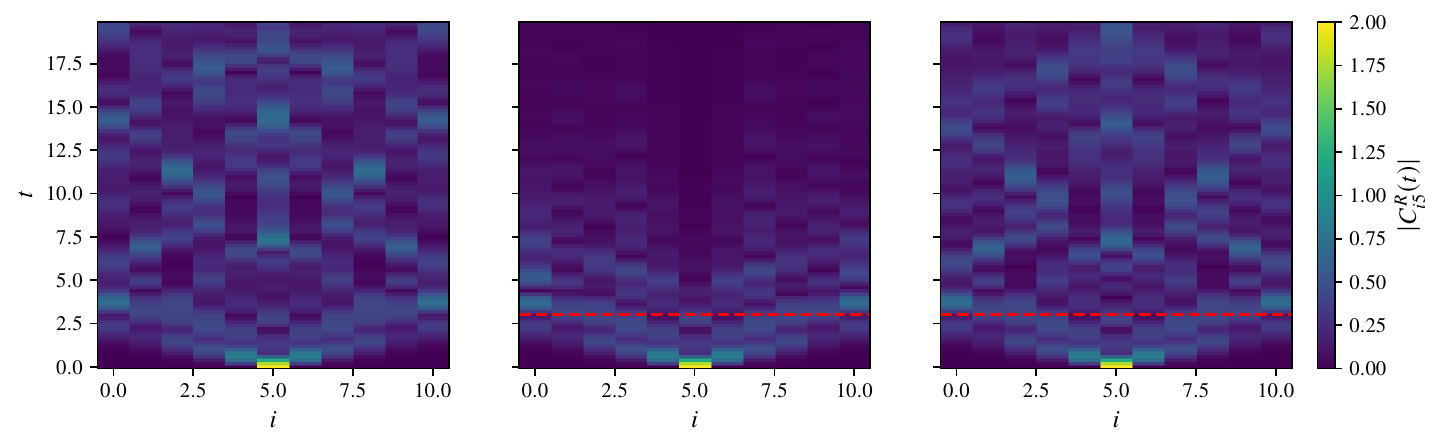}
	\caption{Comparison between exact (left), pairwise meAR (middle), and meVAR  (right) space-time dynamics of the absolute value of $C^R_{i5}(t)$. meAR and meVAR signals have been extended from $t=3.0$ (red dashed line). In meVAR we have included all the $11$ sites of the chain, i.e., all the $121$ correlators. Here, the sampling interval is $\Delta t =0.2$. Note that the pairwise meAR reconstruction (middle) is able to faithfully reproduce the original signal only up to $t \approx 5$, with a relative $L^2$ error $\|C^R_{\rm exact} - C^R_{\rm meAR}\|_2/ \|C^R_{\rm exact}\|_2 = 0.84$ over the space-time region shown in the plot. On the other hand, meVAR is able to accurately capture all the non-trivial features of the exact signal up to $t=20$ (and even beyond), with just a minor overall damping or small shifts in the position of few peaks, with a relative $L^2$ error $\|C^R_{\rm exact} - C^R_{\rm meVAR}\|_2/ \|C^R_{\rm exact}\|_2 = 0.35$ over the space-time region shown in the plot.}	
	\label{fig:CR_spacetime}
\end{figure*}

\begin{figure*}
    \centering
    \includegraphics[width=\textwidth]{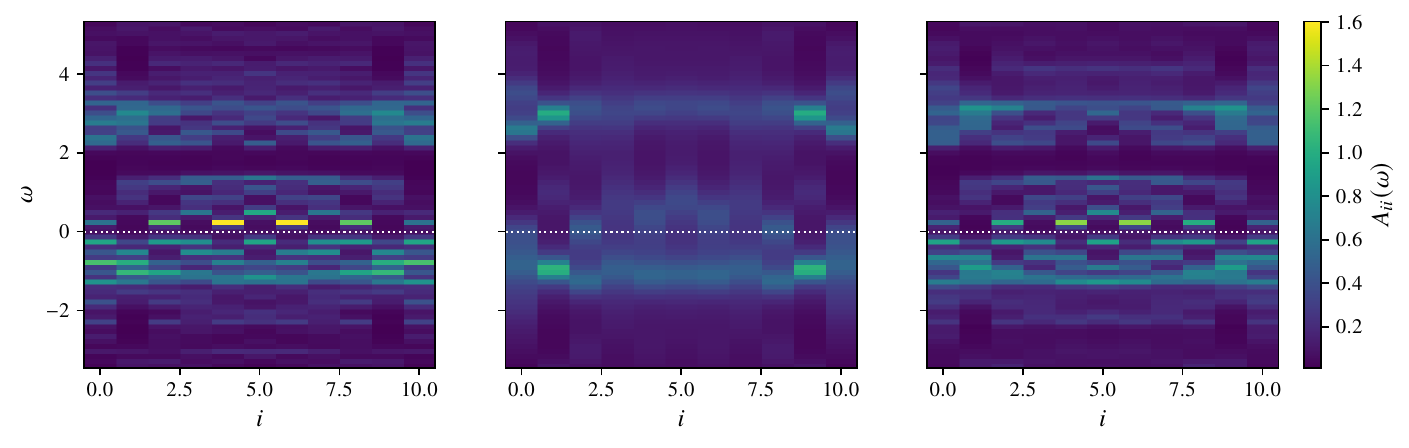}
    \caption{Local density of states $A_{ii}(\omega)$ obtained from the extension of $C^R_{ii}(t)$. Left: exact. Middle: pairwise meAR. Right: $11$-site meVAR. The spectral function has been obtained from the Fourier transform of the exact and extended signals with $t = 50$. Both meAR and meVAR results have been obtained from extensions from $t = 3.0$. Pairwise meAR (middle) provides a smooth reconstruction of the exact LDoS, with a relative $L^2$ error $\|A_{\rm exact} - A_{\rm meAR}\|_2/ \|A_{\rm exact}\|_2 = 0.53$ over the space-frequency region shown in the plot. On the other hand, the LDoS obtained from meVAR accurately reproduces all the features of the exact LDoS, with a relative $L^2$ error $\|A_{\rm exact} - A_{\rm meVAR}\|_2/ \|A_{\rm exact}\|_2 = 0.34$ over the space-frequency region shown in the plot.
    In all panels, the white dotted line represents the Fermi energy and we used a Lorentzian broadening factor $\eta = 0.1$.}
    \label{fig:LDoS}
\end{figure*}

\begin{figure*}
	\centering
	\includegraphics[width=\textwidth]{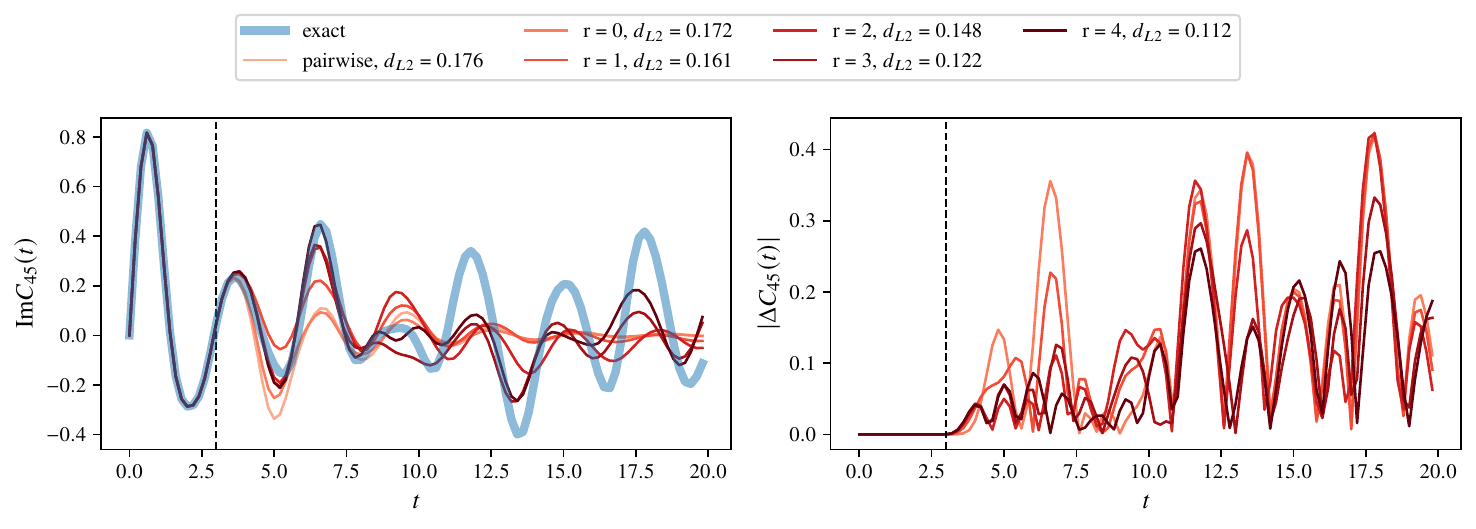}
	\caption{Comparison between the extensions of $C^R_{45}(t)$ obtained via pairwise meAR and meVAR for $C^R_{45}(t)$. Note that $C^R_{45}(t)$ is purely imaginary. In both cases, the correlator has been extended from $t = 3.0$ (vertical dashed line). Left: plot of $\mathrm{Im}C^R_{45}(t)$; Right: absolute error of the extended correlator with respect to the exact one, $\Delta C^R(t) = |C^R_{\rm ext}(t) - C^R_{\rm exact}(t)|$.  Here, the exact dynamics is plotted in blue, while the red curves denote extensions obtained through pairwise meAR and meVAR including sites with an increasing distance $r$ from the initial ones. That is, with $r=0$, $i,j\in\{4,5\}$, with $r=1$, $i,j\in\{3,4,5,6\}$, and so on. For $r=4$, all the $11$ sites (corresponding to 121 correlators) have been included in the VAR model. Pairwise meAR and meVAR with $r=0$ give a similar result. However, the performances of meVAR improves significantly by adding neighboring sites. When all the $11$ sites of the chain are included, the meVAR extension of $C^R_{45}(t)$ accurately reproduces the exact dynamics up to $t \simeq 10$. Here, the sampling interval is $\Delta t =0.2$. In the legend, $d_{L2}$ denotes the $L_2$ distance between the various curves over the shown time interval.}	
	\label{fig:CR_offdiag_vs_radii}
\end{figure*}

In \cref{fig:A_vs_radii}, we compare the local density of states (LDoS) at the center of the chain, $A_{55}(\omega)$, obtained by extending the correlator $C^R_{55}(t)$ using meAR and meVAR with all the sites of the chain. While meAR returns a smooth interpolation of the actual LDoS, meVAR is able to capture its rich structure.  

\begin{figure*}
	\centering
	\includegraphics[width=\textwidth]{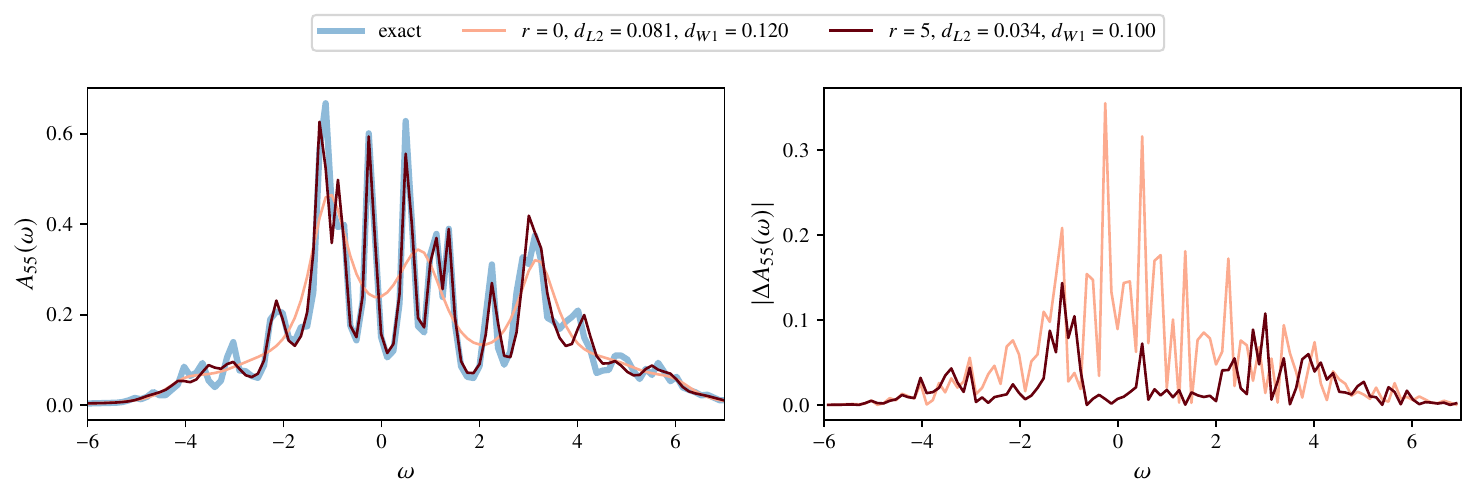}
	\caption{Comparison between local density of states (LDoS) $A_{55}(\omega)$ at the central site $i=5$ obtained using meAR and meVAR. Left: The blue curve represents the exact LDoS, obtained by a Fourier transform of the exact correlator up to  $t=50$. The light (dark) red curve represents the LDoS obtained from a meAR (meVAR with all the sites) extension up to $t = 50$ of $C^R_{55}(t)$ from $t=3.0$. Despite the presence of many sharp features in $A_{55}(\omega)$ and the short training window, the meVAR reconstruction captures the majority of the peaks with exceptional accuracy while the meAR LDoS is a smoothed interpolation of the exact spectrum. Right: Error between the exact and reconstructed LDoS. In the legend, $d_{L2}$ and $d_{W1}$ denote the $L_2$ and Wasserstein \cite{Peyre_2019, Cazelles_2021} distances between the various curves in the left panel, respectively.  Here, we used a Lorentzian broadening factor $\eta = 0.1$.}
	\label{fig:A_vs_radii}
\end{figure*}

In \cref{fig:error_vs_training} we show the relative $L^2$ error between the reconstructed LDoS and the exact one at the center of the chain, $A_{55}(\omega)$, as a function of the training time $t_{\rm train}$ from which the correlator $C_{ij}^R(t)$ is extended. We note that, even for short $t_{\rm train}$, the error of the meVAR extension is significantly smaller than the one of the meAR extension with much longer training times. In particular, the error for meVAR with $r = 5$ (all $11$ sites included) drops below $10\%$ already at $t_{\rm train} = 4$, while larger training times are needed for smaller values of $r$. A clear trade-off thus emerges between the training time and the number of correlators included in the meVAR model: the more correlators are included, the shorter the time evolution needed to obtain an accurate LDoS.  

\begin{figure}
    \centering
    \includegraphics[width=\linewidth]{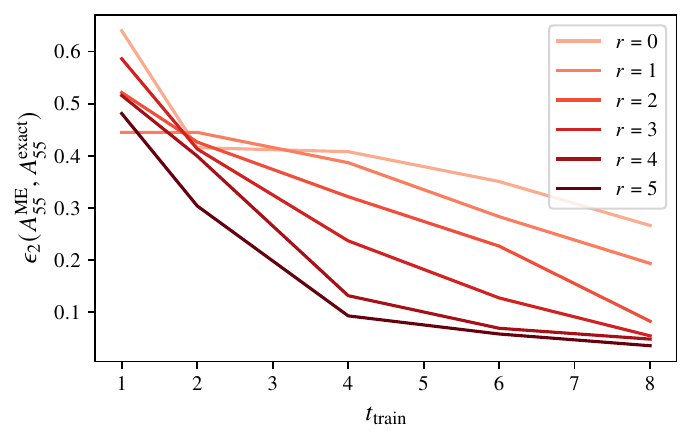}
    \caption{Relative $L^2$ error, $\varepsilon_2(x,y) = \|x-y\|_2/\|y\|_2$, between the exact LDoS at the center of the chain, $A^{\rm exact}_{55}(\omega)$, and the one obtained from extending the retarded correlation function $C^R_{55}(t)$ from $t_{\rm train}$ using meAR ($r = 0$) and meVAR including neighbors with distance $r$ from the central site,  $A^{\rm ME}_{55}(\omega)$. Including more neighbors in the VAR model allows one to achieve significantly lower reconstruction errors and from shorter training times.}
    \label{fig:error_vs_training}
\end{figure}

In \cref{fig:noise}, we investigate the performance of our maximum-entropy extension in the presence of noise. For example, in quantum-computing simulations, such noise may arise finite-shot statistics. Assessing the robustness of time-extension methods to such perturbations is particularly important, as these approaches are generally quite sensitive to noise~\cite{Erpenbeck_2026}.

To emulate the effect of noise, we add to the signal a complex Gaussian perturbation with independent real and imaginary components drawn from $\mathcal{N}(0, \sigma^2_{\rm noise})$. We first draw a single realization to represent a hypothetical noisy trajectory $\hat{C}^R_{ij}(t)$. Using this trajectory as the reference, we then generate $50$ parametric bootstrap realizations $C^{R,(s)}_{ij}(t)$ to estimate the bootstrap confidence intervals.

Since both meAR and meVAR require a PSD covariance sequence as input, each bootstrap instance $C^{R,(s)}_{ij}(t)$ is denoised using a generalization of the PSD alternating projection algorithm introduced in Ref.~\cite{Kemper_2024} to the block-Toeplitz matrix of \cref{eq:Gt_block_gram} (see \cref{app:sec:psd_denoising}). In \cref{fig:noise} we show both the extended correlator and the reconstructed LDoS for the central site $i=5$ obtained via meVAR with all the $11\times11$ correlators in the presence of strong Gaussian noise with $\sigma_{\rm noise} = 0.1$. Note that, in this case, the correlators are extended from $t = 5.0$: in the presence of noise, slightly longer time series are required to obtain accurate results. As we discuss in \cref{app:sec:psd_denoising}, in the presence of noise, it is reasonable to apply a regularized version of the PSD denoising algorithm by imposing $R^{(n)}\succeq \delta I$ for some noise-dependent threshold value $\delta$ instead of the standard $R^{(n)}\succeq 0$. This is particularly important for meAR or meVAR with few correlators, as they can be quite fragile to noise (not shown). In contrast, beyond improving the prediction power and the accuracy of the spectral reconstruction, as we have seen above, increasing the number of correlators included in meVAR also enhances its resilience to noise. To show this, in \cref{fig:noise} we use meVAR with standard, non-regularized block-PSD denoising (i.e., with $\delta = 0$). The latter is extremely efficient in smoothing the signal and emoving most of the noise-induced variance from the data, as can be seen from the comparison between the original signal (green) with the denoised one (orange) for $t \leq t_{\rm train}$. This is essential to achieve a small variance in the subsequent extension step. After denoising, each bootstrap sample $C^{R,(s)}_{ij}(t)$ is extended from $t = 5$ to get the $95\%$ bootstrap confidence interval of the extension. The solid line represents the extension of the denoised best estimator, namely the simulated experimental correlator $\hat{C}^{R}_{ij}(t)$, from $t = 5$. Note that, despite the strength of the noise, both the extended correlator and its spectral function are in good agreement with the exact ones.

\begin{figure*}
	\centering
	\includegraphics[width=\textwidth]{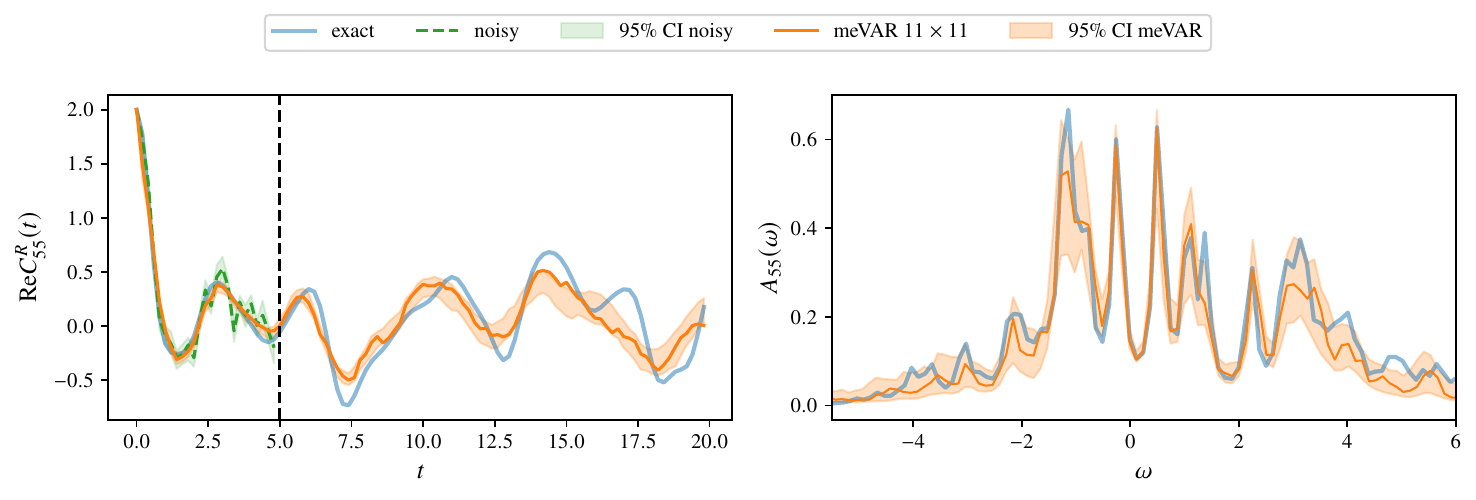}
	\caption{
    Time extension of the real part of a noisy correlator $C^R_{55}(t)$ (left) and corresponding LDoS $A_{55}(\omega)$ (right) obtained via a meVAR extension from $t=5$ including all the $11\times11$ correlators. Here, the noise is sampled from a Gaussian distribution with standard deviation $\sigma_{\rm noise} = 0.1$. The green, dashed line represents the noisy, hypothetical experimental correlator, $\hat{C}^{R}_{55}(t)$, and the shaded green area represents the $95\%$ bootstrap confidence intervals obtained from $B = 50$ bootstrap samples, $C^{R,(s)}_{ij}(t)$. Each $C^{R,(s)}_{ij}(t)$ is then denoised using the non-regularized block-PSD denoising algorithm described in \cref{app:sec:psd_denoising}. The denoised correlator $\hat{C}^R_{ij}(t)$ is subsequently extended using meVAR (solid orange curve). The shaded regions correspond to the central $95\%$ confidence interval obtained by extending each of the individual denoised bootstrap realizations $C^{R,(s)}_{ij}(t)$.}
	\label{fig:noise}
\end{figure*}

Finally, in \cref{fig:noise_sensitivity}, we investigate the effect of noise on the reconstructed LDoS as a function of the number of correlators included in the meVAR model. Here, we show the relative $L^2$ error, $\varepsilon_2$, between the exact and the reconstructed LDoS obtained from extending a noisy $C^R_{ij}(t)$ from $t_{\rm train} = 5$ as a function of $\sigma_{\rm noise}$. To mitigate the effects of noise in a systematic way, all correlators are denoised imposing $R^{(n)}\succeq \delta I$, with the same value of $\delta$ determined via a bootstrap analysis of the scalar correlator as explained in \cref{app:sec:psd_denoising}. \Cref{fig:noise_sensitivity} further confirms the improved noise resilience of the meVAR extension. Even in the presence of strong noise, $\sigma_{\rm noise} \sim10^{-1}$, meVAR allows us to achieve a two-fold reduction in the reconstruction error of the LDoS with respect to scalar meAR, and it converges to the noiseless value for intermediate noise strengths.

\begin{figure}
    \centering
    \includegraphics[width=1\linewidth]{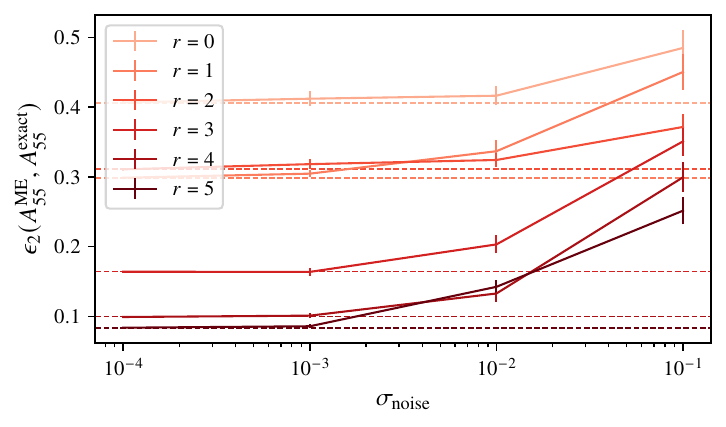}
    \caption{Relative $L^2$ error, $\varepsilon_2(x,y) = \|x-y\|_2/\|y\|_2$, between exact and reconstructed LDoS $A_{55}(\omega)$ as a function of the strength of the Gaussian noise, $\sigma_{\rm noise}$, for meAR ($r = 0$) and meVAR including neighbors with distance $r$ from site $i=5$. The solid line represents a bootstrap mean over $B=50$ instances, with the error bars given by the corresponding standard deviation. For all the values of $r$, we used the PSD denoising algorithm with regularization parameter $\delta$ obtained via a bootstrap analysis of $C^{R}_{55}(t)$ discussed in \cref{app:sec:psd_denoising}. Here, the dashed lines represent the zero-noise asymptotic values corresponding to the various values of $r$. } 
    \label{fig:noise_sensitivity}
\end{figure}

\section{Conclusions}
We have shown that the sampled time series of a quantum correlator can be identified with the covariance sequence of a wide-sense stationary stochastic process, turning the extension of short time series into a covariance extension problem. We have shown that such an extension can be determined by a maximum-differential-entropy-rate criterion. This selects an autoregressive model whose coefficients are obtained from the solution of a set of Yule--Walker equations, and guarantees a bounded extension, a positive spectral function, and an exact total spectral weight sum rule. Since this autoregressive model parameterizes a spectral density whose poles need not correspond to physical excitations but collectively form a rational approximation of the spectral density, it remains faithful even when applied to correlators with a structured or continuous spectral function. This approach generalizes naturally to matrix correlators. Besides enforcing their analytic properties, this formulation captures both temporal auto-correlations within individual components and spatial cross-correlations between them, allowing it to exploit dynamical information shared across multiple correlators. This shared structure significantly enhances the predictive accuracy and the resilience to noise of the method compared with treating each correlator independently. As discussed in \cref{app:sec:psd_denoising}, the approach we discussed in this work might be used to partially mitigate errors arising from imperfect state preparation and Trotterization, at the price of measuring additional correlation functions. We will explore this possibility in future work. We expect this method to be beneficial to both classical and quantum simulations of correlation functions, in particular for those with structured or continuous spectra, enabling spectral functions to be reconstructed with significantly higher accuracy from shorter and potentially noisy time series.

\begin{acknowledgments}
This work was supported by Innovate UK (grant number 10120959). The authors thank P. K. Faehrmann for helpful discussions at an early stage of this work.
The authors used generative AI tools to assist with the literature review, research discussions, coding, language editing, and proofreading. The authors critically reviewed the resulting suggestions and take full responsibility for the scientific content, code, and conclusions of this work. 
\end{acknowledgments}

\bibliography{Bibliography.bib}
\appendix

\section{Stochastic processes}
\label{app:sec:stochastic_processes}
We briefly summarize here the definition and main properties of a stochastic autoregressive process. See Refs.~\cite{BrockwellDavis1991,StoicaMoses2005} for further details.
A discrete stochastic process is a collection of complex random variables $\{X_t\}_{t \in \mathbb{Z}}$ indexed by a parameter $t$, which typically is the time. We assume throughout $\mathbb{E}[|X_t|^2] <\infty$, $\forall t\in\mathbb{Z}$ (a second-order process). The mean and covariance of a stochastic process are
\begin{align}
    m_t &:= \mathbb{E}[X_t],\\
    C(s,t) &:= \mathbb{E}[(X_s-m_s)(X_t-m_t)^*],
\end{align}
respectively. A complex stochastic process is also characterized by the pseudo-covariance $\mathbb{E}[(X_s-m_s)(X_t-m_t)]$ and the process is called proper when the pseudo-covariance vanishes $\forall s,t$. In what follows we will consider only proper stochastic processes. It is convenient to consider the centered process, $X_t \leftarrow X_t -m_t$, such that $\{X_t\}_{t\in\mathbb{Z}}$ has zero mean and covariance $C(s,t) = \mathbb{E}[X_sX^*_t]$. Unless stated otherwise, in this work we consider centered processes. 

A (centered) process is wide-sense stationary (WSS) if its mean is time-independent, $\mathbb{E}[X_t] = 0$ $\forall t\in\mathbb{Z}$, and its covariance only depends on the time difference, $C(t+k,t) = \gamma_k$ $\forall t,k\in\mathbb{Z}$. The sequence $\{\gamma_k\}_{k\in\mathbb{Z}}$, with 
\begin{equation}
    \label{app:eq:autocov_lags}
    \gamma_k := \mathbb{E}[X_{t+k}X^*_t],
\end{equation}
is called the covariance (or autocovariance) sequence. The covariance lags $\gamma_k$ satisfy $\gamma_{-k}=\gamma_k^*$. Assuming for simplicity that $\sum_k |\gamma_k|<\infty$, the covariance lags can be expresses in terms of the spectral density $f(\theta)$ via
\begin{equation}
    \label{app:eq:spectral_density_scalar}
    \gamma_k = \frac{1}{2\pi}\int_{-\pi}^{\pi}f(\theta)e^{ik\theta}d\theta.
\end{equation}

A zero-mean, proper, WSS process $\{X_t\}_{t\in\mathbb{Z}}$ is an autoregressive process of order $p$, denoted by AR($p$), if it satisfies
\begin{equation}
	\label{app:eq:ARp}
	X_t = \sum_{k=1}^{p}a_kX_{t-k} + \varepsilon_t,\qquad\forall t\in\Z,
\end{equation}
where $a_k\in\mathbb{C}$ are the AR coefficients and the innovation $\varepsilon_t\in\mathbb{C}$ is a random variable accounting for the component of $X_t$  that is not explained by the linear combination of the previous $p$ values. We assume innovations white and proper, i.e., $\mathbb{E}[\varepsilon_t] = 0$ and $\mathbb{E}[\varepsilon_t\varepsilon_s^*] = \sigma_{\varepsilon}^2\delta_{ts}$ and $\mathbb{E}[\varepsilon_t\varepsilon_s] = 0$, $\forall s, t\in\mathbb{Z}$, with $\sigma^2_\varepsilon$ the innovation variance of the AR($p$) process. We further assume that $\varepsilon_t$ is orthogonal to the past of the process, i.e., $\mathbb{E}[\varepsilon_tX_s^*] = 0, \forall s < t$. 

The spectral density of an AR($p$) process is given by~\cite{StoicaMoses2005}
\begin{equation}
    \label{app:eq:spectral_arp}
    f(\theta) = \frac{\sigma^2_\varepsilon}{\left|1-\sum_{k=1}^{p}a_ke^{-ik\theta}\right|^2},
\end{equation}
with $A(z) = 1 - \sum_{k=1}^{p}a_k z^k$ the AR polynomial. For a WSS process with $\sigma_\varepsilon^2>0$, orthogonality of the innovations to the past implies $A(z)\neq0$ for $|z|\leq1$~\cite{StoicaMoses2005, BrockwellDavis1991}.

In linear prediction theory, given a centered WSS stochastic process $\{X_t\}_{t\in\mathbb{Z}}$, the optimal linear one-step predictor of order $p$ is
\begin{equation}
	\label{app:eq:AR_prediction}
	\hat{X}_t := \sum_{k=1}^{p}a_kX_{t-k},
\end{equation}
where the coefficients $\{a_k\}_{k=1}^{p}$ minimize the mean-square prediction error $\mathbb{E}[|e_t|^2]$, with
\begin{equation}
	\label{app:eq:AR_prediction_error}
	e_t := X_t - \sum_{k=1}^{p}a_kX_{t-k}.
\end{equation}
By the orthogonality principle, the minimum is attained if and only if $\mathbb{E}[e_tX_{t-k}^*] = 0$ for $k=1,\dots,p$. Substituting \cref{app:eq:AR_prediction_error} into this condition and using the definition of the covariance lags in \cref{app:eq:autocov_lags}, we obtain a set of $p$ equations, known as the Yule--Walker (YW) equations,
\begin{equation}
    \label{app:eq:yule_walker}
    	\gamma_k = \sum_{j=1}^{p}a_j \gamma_{k-j}, \qquad k = 1,\dots,p,
\end{equation}
which relate the optimal coefficients $\{a_k\}$ to the covariance lags $\{\gamma_k\}_{k=0}^{p}$.

Using $\gamma_{-k} = \gamma^*_k$, \cref{app:eq:yule_walker} can be written in matrix form as
\begin{equation}
	\label{eq:YW_matrix_form}
	R^{(p)}\bm{a} = \bm{g}, 
\end{equation}
where we introduced the vectors $\bm{g} := (\gamma_1,\dots,\gamma_p)^T$, $\bm{a} := (a_1, \dots, a_p)^T$, and the Hermitian PSD Toeplitz matrix
\begin{equation}
\label{app:eq:covariance_gram_matrix}
	R^{(p)}  := \begin{bmatrix}
		\gamma_0 & \gamma_1^*  & \cdots & \gamma_{p-1}^*\\
		\gamma_1 & \gamma_0 & \cdots & \gamma_{p-2}^*\\
		\vdots & \vdots  & \ddots  & \vdots\\
		\gamma_{p-1} & \gamma_{p-2} & \cdots & \gamma_0\\
	\end{bmatrix}.
\end{equation}

If $R^{(p)}\succ 0$, the optimal coefficients $\{a_k\}_{k=1}^{p}$ are uniquely given by
\begin{equation}
	\bm{a} = [R^{(p)}]^{-1}\bm{g}.
\end{equation}
Once the optimal coefficients have been determined, the prediction-error variance, $\sigma_{p}^2 = \mathbb{E}[|e_t|^2]$ follows from $X_t = \hat{X}_t + e_t $ and $\mathbb{E}[e_t\hat{X}_t^*] = 0$,
\begin{equation}
    \label{app:eq:yw_sigma}
    \sigma_{p}^2 = \gamma_0 - \sum_{k=1}^{p}a_k \gamma_k^* = \gamma_0 -\bm{g}^\dagger\bm{a},
\end{equation}
which is real and satisfies $0\leq \sigma_p^2\leq \gamma_0$. Note that, if $\{X_t\}$ is an AR($p$) process, the optimal coefficients coincide with those in \cref{app:eq:ARp}, $e_t=\varepsilon_t$ and $\sigma_p^2=\sigma_\varepsilon^2$.

For a general centered WSS stochastic process, the solution to the YW equations \cref{app:eq:yule_walker} gives the minimum-mean-square linear predictor of $X_t$ based on its previous $p$ values. The prediction error $e_t$ is guaranteed to satisfy $\E[e_tX_{t-k}^*] =0 $ only for $k=1,\dots,p$. If instead the process is an AR($p$) process, then $\varepsilon_t = e_t$, $\sigma_p^2 = \sigma_{\varepsilon}^2$, and $\mathbb{E}[\varepsilon_t X_{t-k}^*] = 0$, $\forall k\geq 1$. In this case, 
the YW recursion $\gamma_k=\sum_{j=1}^pa_j\gamma_{k-j}$ holds $\forall k\geq 1$. Together with $\gamma_{-k}=\gamma_k^*$, this makes it possible to determine the whole covariance sequence $\{\gamma_k\}_{k\in\Z}$ from $\{\gamma_k\}_{k=0}^{p}$. More generally, if $R^{(p)}\succ0$ and $\sigma_p^2 > 0$, the YW solution defines an AR($p$) process whose first $p+1$ covariance lags coincide with the given ones~\cite{BrockwellDavis1991}. Hence any such finite sequence can be extended recursively to a valid covariance sequence on $\mathbb{Z}$, although for a non-AR process the extended lags $\gamma_k$, $k>p$, differ in general from the true ones.

\section{Vector stochastic processes}
\label{app:sec:var}

Here we generalize what we discussed in \cref{app:sec:stochastic_processes} to vector (or multichannel) stochastic processes; see Refs.~\cite{BrockwellDavis1991,Lutkepohl2005} for further details. A discrete $M$-dimensional complex stochastic process is a collection $\{\bm{X}_t\}_{t\in\mathbb{Z}}$ of random vectors $\bm{X}_t = (X_{1,t},\dots, X_{M,t})^T\in\C^{M}$. We assume $\mathbb{E}[\|\bm X_t\|_2^2]<\infty$  $\forall t\in\mathbb{Z}$, where $\|\cdot\|_2$ is the Euclidean norm. Its mean vector and covariance matrix are given by
\begin{align}
    \bm{m}_t &:= \mathbb{E}[\bm{X}_t],\\
    C(s,t) &:= \E[(\bm{X}_s - \bm{m}_s)(\bm{X}_t - \bm{m}_t)^\dagger],
\end{align}
respectively. Here, $(\bm X_s-\bm m_s)(\bm X_t-\bm m_t)^\dagger$ is an outer product, so that $C(s,t)$ is an $M\times M$ matrix with entries
\begin{equation}
    [C(s,t)]_{ij} = \E[(X_{i,s } - m_{i, s})(X_{j,t} - m_{j,t})^*].
\end{equation}
As in the scalar case, the process is called proper if its pseudocovariance $\E[(\bm{X}_s - \bm{m}_s)(\bm{X}_t - \bm{m}_t)^T]$ vanishes $\forall s,t \in \Z$. We consider only proper, centered processes. A centered vector process is WSS if it has zero mean, $\E[\bm{X}_t]=\bm{0}$, and $C(t+k,t) = \Gamma_k$, $\forall t, k \in \Z$. Its matrix covariance lags are
\begin{equation}
    \label{app:eq:autocovariance_matrices}
    \Gamma_k := \E[\bm{X}_{t+k}\bm{X}_t^\dagger],
\end{equation}
with $\Gamma_{-k} = \Gamma_k^\dagger$. Assuming for simplicity that $\sum_k\|\Gamma_k\|_2<\infty$, the covariance lags can be expressed in terms of the spectral density matrix $F(\theta)$ via
\begin{equation}
    \label{app:eq:spectral_density_matrix}
    \Gamma_k = \frac{1}{2\pi}\int_{-\pi}^{\pi}F(\theta)e^{ik\theta}d\theta.
\end{equation}

A zero-mean, proper, WSS vector process $\{\bm X_t\}_{t\in\mathbb{Z}}$ is a vector autoregressive process of order $p$, denoted VAR($p$), if it satisfies
\begin{equation}
\label{app:eq:VARp}
	\bm{X}_t = \sum_{k=1}^{p}A_k\bm{X}_{t-k} + \bm{\varepsilon}_t,\qquad\forall t\in\Z,
\end{equation}
where $A_k\in\C^{M\times M}$ are the autoregressive coefficient matrices and  $\bm{\varepsilon}_t\in\C^M$ is the innovation vector. We assume white and proper innovations, i.e., $\E[\bm{\varepsilon}_t]=\bm{0}$ and $\E[\bm{\varepsilon}_t\bm{\varepsilon}_s^\dagger]=\Sigma_{\varepsilon}\delta_{ts}$, and $\E[\bm{\varepsilon}_t\bm{\varepsilon}_s^T]=0$, with $\Sigma_{\varepsilon}\succ0$ the innovation covariance matrix. As in the scalar case, we assume that $\bm\varepsilon_t$ is orthogonal to the past of the process, $\mathbb{E}[\bm\varepsilon_t\bm X_s^\dagger]=0$, $\forall s<t$. Component-wise, \cref{app:eq:VARp} reads
\begin{equation}
	X_{i,t} = \sum_{k=1}^{p}\sum_{j=1}^M(A_k)_{ij}X_{j,t-k} + \varepsilon_{i,t},\qquad\forall t\in\Z,
\end{equation}
that is, in a VAR($p$) process the $M$ components of $\bm{X}_t$ are coupled by the off-diagonal entries of $\{A_k\}$: hence, the model naturally captures both auto- and cross-correlations. 

The spectral density of a VAR($p$) process is given by~\cite{StoicaMoses2005}
\begin{equation}
    \label{app:eq:spectral_varp}
    F(\theta) = [A(e^{-i\theta})]^{-1}\Sigma_{\varepsilon}[A(e^{-i\theta})^\dagger]^{-1},
\end{equation}
with $F(\theta)$ Hermitian and positive definite $\forall \theta$, and $A(z) = I-\sum_{k=1}^{p}A_kz^{k} $ the matrix AR polynomial. For a WSS process with $\Sigma_\varepsilon\succ0$, orthogonality of the innovations to the past implies
$\det A(z)\neq0$ for $|z|\leq1$~\cite{Lutkepohl2005}. 

For a general zero-mean WSS vector process, the optimal linear one-step predictor of order $p$ is
\begin{equation}
    \label{app:eq:VAR_prediction}
	\hat{\bm{X}}_t := \sum_{k=1}^{p}A_k\bm{X}_{t-k},
\end{equation}
with prediction error
\begin{equation}
    \label{app:eq:prediction_error_vector}
    \bm{e}_t := \bm{X}_t -\sum_{k=1}^{p}A_k\bm{X}_{t-k}.
\end{equation}

As in the scalar case, the optimal coefficient matrices $\{A_k\}_{k=1}^{p}$ minimizes $\E[\|\bm{e}_t\|_2^2]$. By the orthogonality principle, this holds if and only if $\mathbb{E}[\bm{e}_t\bm{X}_{t-k}^\dagger]=0$ for $k=1,\dots,p$. Note that each of these $p$ equations is an $M\times M$ matrix equation, $\E[e_{i,t}X_{j,t-k}^*] = 0$ $\forall i,j$. Substituting \cref{app:eq:prediction_error_vector} and using the definition of the covariance matrices in \cref{app:eq:autocovariance_matrices}, we obtain the block (or multichannel) YW equations 
\begin{equation}
    \label{app:eq:block_yw}
    \Gamma_k = \sum_{j=1}^{p}A_j \Gamma_{k-j}, \qquad k=1,\dots,p.
\end{equation}
Defining the \emph{row} vectors of matrices $\bm{\mathcal{A}} := (A_p,\dots,A_1)$ and $\bm{\mathcal{G}} := (\Gamma_p,\dots,\Gamma_1)$, the equation above can be written in compact form as
\begin{equation}
	\label{app:eq:matrix_yw}
	\bm{\mathcal{A}} \mathcal{R}^{(p)}= \bm{\mathcal{G}},
\end{equation}
with the following Hermitian PSD block-Toeplitz matrix
\begin{equation}
    \label{app:eq:covariance_gram_multi}
	\mathcal{R}^{(p)} :=\begin{bmatrix}
		\Gamma_0 & \Gamma^\dagger_{1}  & \cdots & \Gamma^\dagger_{p-1}\\
		\Gamma_1 & \Gamma_0 & \cdots & \Gamma^\dagger_{p-2}\\
		\vdots & \vdots  & \ddots  & \vdots\\
		\Gamma_{p-1} & \Gamma_{p-2} & \cdots & \Gamma_0\\
	\end{bmatrix}.
\end{equation}
If $\mathcal{R}^{(p)}\succ0$, the optimal coefficient matrices $\{A_k\}$ are uniquely given by
\begin{equation}
\label{app:eq:block_yw_solution}
\bm{\mathcal{A}} = \bm{\mathcal{G}}[\mathcal{R}^{(p)}]^{-1}.
\end{equation}

The prediction-error covariance $\Sigma_p = \E[\bm{e}_t\bm{e}_t^\dagger]$ follows from $\mathbb{E}[\bm e_t\hat{\bm X}_t^\dagger]=0$, which gives $\Sigma_p=\mathbb{E}[\bm e_t\bm X_t^\dagger]$, and hence
\begin{equation}
    \label{app:eq:Sigma_multi}
    \Sigma_p = \Gamma_0 -\sum_{k=1}^{p}A_{k}\Gamma_k^\dagger.
\end{equation}

As in the scalar case, for a general zero-mean WSS vector process the solution of the block YW equations \cref{app:eq:block_yw} gives the minimum-mean-square linear predictor of $\bm{X}_t$ based on its previous $p$ values, and the prediction error is guaranteed to satisfy $\mathbb{E}[\bm{e}_t\bm{X}_{t-k}^\dagger]=0$ only for $k=1,\dots,p$. If instead the process is VAR($p$), then $\bm{e}_t=\bm{\varepsilon}_t$, $\Sigma_p=\Sigma_{\varepsilon}$, and $\mathbb{E}[\bm{\varepsilon}_t\bm{X}_{t-k}^\dagger]=0$, $\forall k\ge1$. In this case, the block YW recursion $\Gamma_k=\sum_{j=1}^pA_j\Gamma_{k-j}$ then holds $\forall k\ge1$. Together with $\Gamma_{-k}=\Gamma_k^\dagger$, this determines the whole covariance sequence $\{\Gamma_k\}_{t\in\Z}$ from $\{\Gamma_k\}_{k=0}^{p}$. More generally, if $\mathcal{R}^{(p)}\succ0$ and $\Sigma_p\succ0$, the block YW solution defines a VAR($p$) process whose first $p+1$ covariance lags coincide with the given ones. Any such finite sequence can therefore be extended recursively to a valid covariance sequence on $\mathbb{Z}$, although for a non-VAR process the extended lags $\Gamma_k$, $k>p$, differ in general from the true ones. 

\section{Maximum entropy rate stochastic processes}
\label{app:sec:entropy_rate}
In this section we review Burg's maximum entropy principle. We follow the approach of Refs.~\cite{Choi_1987, CoverThomas2006} adapted to complex random variables~\cite{Neese_1993}. 

\subsection{Differential entropy for complex random variables}
The differential entropy of a complex random variable $Z$ with probability density function $f(z)$, $z\in\C$, is
\begin{equation}
    h(Z) := -\int_{\C} f(z) \log f(z) d^2z,
\end{equation}
where $d^2z = d(\operatorname{Re}z)\,d(\operatorname{Im}z)$ and $\log = \log_2$. For any complex random variable $Z$ with mean $\mu = \E[Z]$ and variance $\sigma^2 = \E[|Z-\mu|^2]$, one can show (using, e.g., the Kullback--Leibler divergence~\cite{Choi_1987,CoverThomas2006,Neese_1993}), that 
\begin{equation}
    h(Z) \leq \log(\pi e\sigma^2),
\end{equation}
with  equality if and only if $Z$ is a proper complex Gaussian random variable, $Z\sim \mathcal{CN}(\mu, \sigma^2)$. Recall that $Z$ is proper if its pseudo-covariance $\E[(Z-\mu)^2]$ vanishes. Writing $Z-\mu = X+iY$, with $X\sim\mathcal{N}(0,\sigma^2_X)$ and $Y\sim\mathcal{N}(0,\sigma^2_Y)$, this is equivalent to $\sigma^2_X = \sigma^2_Y = \sigma^2/2$ and $\mathrm{Cov}(X,Y) = 0$. The probability density function of a proper complex Gaussian random variable is  
\begin{equation}
    p(z) = \frac{1}{\pi\sigma^2}e^{-|z-\mu|^2/\sigma^2},
\end{equation}
with differential entropy 
\begin{equation}
    h(Z) = \log(\pi e\sigma^2).
\end{equation}

For a complex random vector $\bm{Z} = (Z_1, \dots, Z_n)^T$, with joint probability density $f(\bm{z}) = f(z_1,\dots,z_n)$, the differential entropy is
\begin{equation}
    h(\bm{Z}) := -\int_{\C^n}f(z_1,\dots,z_n)\log f(z_1,\dots,z_n) d^{2n}z.
\end{equation}
As in the single random variable case, we have
\begin{equation}
    h(\bm{Z}) \leq \log\det(\pi e K),
\end{equation}
with covariance matrix $K = \E[(\bm{Z}-\bm{\mu})(\bm{Z}-\bm{\mu})^\dagger]\succ 0$, and $\bm{\mu} = \E[\bm{Z}]$. Equality holds if and only if $\bm{Z}$ is a proper complex Gaussian random vector, $\bm{Z}\sim\mathcal{CN}(\bm{\mu},K)$, i.e., $\E[(\bm{Z}-\bm{\mu})(\bm{Z}-\bm{\mu})^T] = 0$. In this case, the probability density function is 
\begin{equation}
    p(\bm{z}) = \frac{1}{\pi^{n}\det K}e^{-(\bm{z}-\bm{\mu})^\dagger K^{-1} (\bm{z}-\bm{\mu})}
\end{equation}
and its differential entropy
\begin{equation}
    \label{app:eq:diff_entropy_gaussian_vector}
    h(\bm{Z}) = \log[(\pi e)^n\det K].
\end{equation}

\subsection{Differential entropy rate of a stochastic process}
The differential entropy rate of a stochastic process $\Xt$ is defined as
\begin{equation}
    \label{app:eq:diff_entropy_rate}
    \bar{h}(\Xt):= \lim_{n\rightarrow\infty}\frac{h(X_1,\dots,X_n)}{n}.
\end{equation}
For a stationary process this limit exists. By the chain rule, 
\begin{equation}
    \label{app:eq:chain_rule}
    h(X_1,\dots,X_n) = \sum_{i = 1}^{n}h(X_i|X_{i-1},\dots,X_1).
\end{equation}
By stationarity, and since conditioning cannot increase entropy, the terms $h(X_n|X_{n-1},\dots,X_1)$ are non-increasing in $n$.
Their limit therefore exists and equals the Cesàro mean \cref{app:eq:diff_entropy_rate}:
\begin{equation}
    \label{app:eq:entropy_rate_conditional}
    \bar{h}(\Xt) = \lim_{n\rightarrow\infty}h(X_n|X_{n-1},\dots,X_1).
\end{equation}
Recall that a stochastic process $\Xt$ with mean $m_t = \E[X_t]$ is a proper complex Gaussian process if every finite-dimensional vector $(X_{t_1},\dots,X_{t_n})^T$ is jointly complex Gaussian and
\begin{equation}
    \E[(X_t-m_t)(X_s-m_s)] = 0, \ \forall t,s.
\end{equation}
Since the entropy is invariant under translations, we can take the process to be centered. For $\bm X=(X_1,\dots,X_n)^T$, \cref{app:eq:diff_entropy_gaussian_vector} gives 
\begin{equation}
    h(X_1,\dots,X_n) = \log[(\pi e)^n \det K^{(n)}],
\end{equation}
with covariance matrix $K^{(n)} = \E[\bm{X}\bm{X}^\dagger]\succ 0$. Then, the differential entropy rate of $\Xt$ is given by
\begin{equation}
    \bar{h}(\Xt) = \log(\pi e) +\lim_{n\rightarrow\infty}\frac{1}{n}\log\det K^{(n)}.
\end{equation}
For a stationary process, $K^{(n)}$ is a Hermitian Toeplitz matrix with $K^{(n)}_{ij} = \gamma_{i-j}$ and $\gamma_k = \E[X_{t+k}X_t^*]$, i.e., $K^{(n)}=R^{(n)}$ in the notation of \cref{app:eq:covariance_gram_matrix}. Using \cref{app:eq:spectral_density_scalar} and Szeg\H{o}'s theorem~\cite{GrenanderSzego1984,Gray_2006}, we obtain
\begin{equation}
    \label{app:eq:gaussian_entropy_rate_scalar}
    \bar{h}(\Xt) = \log(\pi e) + \frac{1}{2\pi}\int_{-\pi}^{\pi}\log f(\theta)d\theta.
\end{equation}

Similarly, for an $M$-dimensional proper complex stationary Gaussian process $\{\bm{X}_t\}$ the block version of Szeg\H{o}'s theorem gives
\begin{equation}
    \bar{h}(\{\bm{X}_t\}) = M\log(\pi e) + \frac{1}{2\pi}\int_{-\pi}^{\pi}\log \det F(\theta)d\theta,
\end{equation}
with $F(\theta)$ the spectral density matrix defined in \cref{app:eq:spectral_density_matrix}. 

\subsection{Burg's theorem}
\label{app:sec:burg}
We have the following
\begin{theorem}[Burg, scalar case]
\label{app:thm:burg}
Let $\gamma_0,\dots,\gamma_p\in\C$ be a finite PSD sequence such that $R^{(p)}\succ0$ and $\sigma_p^2>0$, with $R^{(p)}$ and $\sigma_p^2$ defined in \cref{app:eq:covariance_gram_matrix,app:eq:yw_sigma}. Among all complex-valued stochastic processes $\Xt$ satisfying
\begin{equation}
    \label{app:eq:burg_constraints}
    \E[X_{t+k}X_t^*] = \gamma_k,\quad k =0,\dots,p,\quad\forall t\in\Z,
\end{equation}
the differential entropy rate is maximized by the proper complex Gaussian AR process
\begin{equation}
    \label{app:eq:burg_gaussian_arp}
    X_t = \sum_{k=1}^{p}a_kX_{t-k}+ \varepsilon_t,
\end{equation}
with $\varepsilon_t\sim\mathcal{CN}(0,\sigma_p^2)$ a proper complex Gaussian innovation independent of $\{X_s\}_{s<t}$. The AR coefficients $
a_1,\dots,a_p\in\C$ solve the YW equations \cref{app:eq:yule_walker}.
\end{theorem}

The AR process in \cref{app:eq:burg_gaussian_arp} indeed reproduces the constraints of \cref{app:eq:burg_constraints}. Multiplying \cref{app:eq:burg_gaussian_arp} by $X_{t-k}^*$ and taking expectations gives \cref{app:eq:yule_walker} for $k=1,\dots,p$, since $\mathbb{E}[\varepsilon_tX_{t-k}^*]=0$, and \cref{app:eq:yw_sigma} for $k=0$. Moreover, the conditions $R^{(p)}\succ0$ and $\sigma_p^2>0$ ensure the process of \cref{app:eq:burg_gaussian_arp} exists and is stationary. If, instead, $R^{(p)}$ is singular or $\sigma_p^2=0$, let $p'=\max\{k\le p:R^{(k)}\succ0\}$. Then $\sigma_{p'}^2=0$, and every process satisfying \cref{app:eq:burg_constraints} has entropy rate $-\infty$. The covariance sequence then has a unique extension, corresponding to a sum of $p'$ undamped complex exponentials. Note also that \cref{app:thm:burg} does not assume $\Xt$ to be zero-mean, Gaussian, or WSS: the constraints fix only second moments at lags up to $p$, and, as the proof shows, nothing else enters the bound.

\begin{proof}
    We follow Refs.~\cite{Choi_1987, CoverThomas2006}. Given a finite sequence $X_1, \dots,X_n$ of an arbitrary stochastic process satisfying \cref{app:eq:burg_constraints} and a finite sequence $Z_1,\dots,Z_n$ of a proper complex Gaussian process with the same covariance matrix as $X_1, \dots,X_n$, we have
     \begin{align*}
         h(X_1,\dots,X_n) &\leq h(Z_1,\dots,Z_n)\\
         &\leq h(Z_1,\dots,Z_p) + \sum_{i=p+1}^{n}h(Z_i|Z_{i-1},\dots,Z_{i-p}),
     \end{align*}
    where we used that the differential entropy of a process with a fixed covariance is maximized by a proper Gaussian random vector, the chain rule \cref{app:eq:chain_rule}, and the fact that conditioning does not increase entropy. Let $Z_1',\dots,Z_n'$ be a proper complex Gaussian AR($p$) process with the same covariances as $Z_1,\dots,Z_n$ up to lag $p$ and of the form of \cref{app:eq:burg_gaussian_arp}, $Z_t' = \sum_{k=1}^{p}a_kZ_{t-k}'+\varepsilon_t$ with $\varepsilon_t\sim\mathcal{CN}(0,\sigma^2_p)$. Then
    \begin{align*}
        h(X_1,\dots,X_n)&\leq h(Z_1,\dots,Z_p) + \sum_{i=p+1}^{n}h(Z_i|Z_{i-1},\dots,Z_{i-p}) \\
        &= h(Z_1',\dots,Z_p') + \sum_{i=p+1}^{n}h(Z_i'|Z_{i-1}',\dots,Z_{i-p}')\\
        & = h(Z_1',\dots,Z_p') + (n-p)\log(\pi e \sigma^2_p),
    \end{align*}
    where we used $h(\sum_{k=1}^p a_k Z'_{i-k} + \varepsilon_i|Z_{i-1}',\dots,Z_{i-p}') = h(\varepsilon_i|Z_{i-1}',\dots,Z_{i-p}')$, which follows from the conditional translational invariance property $h(X+g(Y)|Y) = h(X|Y)$, 
    together with $h(\varepsilon_i|Z_{i-1}',\dots,Z_{i-p}') = h(\varepsilon_i)$, which follows from the independence of $\varepsilon_i$ from the past, and $h(\varepsilon_i) = \log(\pi e \sigma^2_p)$. Finally, for the differential entropy rate we have
    \begin{align*}
        \lim_{n\rightarrow\infty}\frac{h(X_1,\dots,X_n)}{n}&\leq\lim_{n\rightarrow\infty}
        \left[\frac{h(Z_1',\dots,Z_p')}{n} + \frac{n-p}{n}\log(\pi e \sigma^2_p)\right]\\
        &=\log(\pi e \sigma^2_p)
    \end{align*}
\end{proof}

The proof above clarifies why the maximum entropy process is often called the least committal one. Among all processes compatible with the prescribed covariance lags, $h(X_t|\text{past})$ is bounded by the entropy of a Gaussian with the order-$p$ prediction-error variance $\sigma_p^2$. The bound is saturated only when $X_t$ depends on the past solely through its linear predictor, with a Gaussian prediction error independent of the past. In this sense the Burg process adds no structure beyond what the constraints impose.

\Cref{app:thm:burg} generalizes to vector stochastic processes in a straightforward way. We have the following 
\begin{theorem}[Burg, vector case]
\label{app:thm:burg_multi}
Let $\Gamma_0,\dots,\Gamma_p\in\C^{M\times M}$ be a finite matrix-valued PSD sequence such that $\mathcal R^{(p)}\succ0$ and $\Sigma_p\succ0$, with $\mathcal R^{(p)}$ and $\Sigma_p$ defined in \cref{app:eq:covariance_gram_multi,app:eq:Sigma_multi}. Among all complex-valued vector stochastic processes $\{\bm{X}_t\}$ satisfying
\begin{equation}
    \label{app:eq:burg_constraints_multi}
    \E[\bm{X}_{t+k}\bm{X}_t^\dagger] = \Gamma_k,\quad k =0,\dots,p,
\end{equation}
the differential entropy rate is maximized by the proper complex Gaussian VAR process
\begin{equation}
    \label{app:eq:burg_gaussian_varp}
    \bm{X}_t = \sum_{k=1}^{p}A_k\bm{X}_{t-k}+ \bm{\varepsilon}_t,
\end{equation}
with $\bm{\varepsilon}_t\sim\mathcal{CN}(0,\Sigma_p) $ a proper complex Gaussian vector innovation independent of $\{\bm X_s\}_{s<t}$,. The coefficient matrices $
A_1,\dots,A_p \in \C^{M\times M}$ are given by the block YW solution \cref{app:eq:block_yw_solution}.
\end{theorem}

As in the scalar case, multiplying \cref{app:eq:burg_gaussian_varp} on the right by $\bm X_{t-k}^\dagger$ and taking expectations shows that the VAR process reproduces the constraints in \cref{app:eq:burg_constraints_multi}. The conditions $\mathcal{R}^{(p)}\succ0$ and $\Sigma_p \succ 0$ ensure the process of \cref{app:eq:burg_gaussian_varp} exists and is stationary. If instead $\mathcal R^{(p)}$ or $\Sigma_p$ are singular, every process satisfying \cref{app:eq:burg_constraints_multi} again has entropy rate $-\infty$. Unlike in the scalar case, however, this does not in general correspond simply to a reduction of the autoregressive order. In particular, with $\mathcal R^{(p)}\succ0$, $\Sigma_p$ may be singular without being zero. For any $\bm v\in\ker\Sigma_p$, $\mathbb{E}[|\bm v^\dagger\bm e_t|^2]=\bm v^\dagger\Sigma_p\bm v=0$, so the linear combination $\bm v^\dagger\bm X_t$ is perfectly predictable from the previous $p$ values, while combinations outside $\ker\Sigma_p$ retain a non-zero innovation. This case lies outside the VAR model of \cref{app:eq:VARp}, where we assumed $\Sigma_\varepsilon\succ0$.

\section{Connection with spectral density functional}
\label{app:sec:spectral_density}
In Burg's original work~\cite{Burg1975}, the maximum-entropy problem was formulated directly in terms of the spectral density. By \cref{app:eq:gaussian_entropy_rate_scalar}, for a proper complex Gaussian WSS process with spectral density $f(\theta)$, maximizing the entropy rate subject to \cref{app:eq:burg_constraints} is equivalent to
\begin{equation}
\label{app:eq:burg_spectral_problem}
\max_{f\ge0}\ \frac{1}{2\pi}\int_{-\pi}^{\pi}\log f(\theta)d\theta
\end{equation} 
subject to
\begin{equation}
\label{app:eq:spectral_constraints}
\frac{1}{2\pi}\int_{-\pi}^{\pi}f(\theta)e^{ik\theta}d\theta=\gamma_k,\quad k=0,\dots,p.
\end{equation} 
Here we show that this formulation leads to the same result as \cref{app:thm:burg}. Note that the formulation of \cref{app:sec:burg} is more general, as it requires neither Gaussianity nor wide-sense stationarity, nor the existence of a spectral density.

Let $f(\theta)$ be any feasible spectral density satisfying \cref{app:eq:spectral_constraints}, and let $a_1,\dots,a_p$ be the YW coefficients of \cref{app:eq:yule_walker}. Since the YW equations involve only $\gamma_0,\dots,\gamma_p$, these coefficients, and the corresponding $\sigma_p^2$ in \cref{app:eq:yw_sigma}, are the same for every feasible $f(\theta)$. 
The prediction error \cref{app:eq:AR_prediction_error} can be written as
\begin{equation}
    e_t = X_t - \sum_{k=1}^{p}a_kX_{t-k} = \sum_{k=0}^{p}b_k X_{t-k},
\end{equation}
with $b_0 = 1$ and $b_k = -a_k$ for $k=1,\dots,p$. . The covariance lags of $
\{e_t\}$ are
\begin{equation}
    \gamma^{(e)}_k = \E[e{t+k}e^*]=\sum_{r,s=0}^{p}b_r b_s^*\gamma_{k+s-r}.
\end{equation}
Using \cref{app:eq:spectral_density_scalar} for $\gamma_{k+s-r}$, we obtain
\begin{align}
    \gamma^{(e)}_k &= \frac{1}{2\pi} \int_{-\pi}^{\pi}f(\theta) e^{ik\theta}\left|\sum_{r=0}^p b_r e^{-ir\theta}\right|^2d\theta\nonumber\\
    &= \frac{1}{2\pi} \int_{-\pi}^{\pi}f(\theta) e^{ik\theta}\left|A(e^{-i\theta})\right|^2d\theta\nonumber\\
    &= \frac{1}{2\pi} \int_{-\pi}^{\pi} f_{e}(\theta)e^{ik\theta}d\theta,
\end{align}
where
\begin{equation}
    f_{e}(\theta) = \left|A(e^{-i\theta})\right|^2 f(\theta)
\end{equation}
is the spectral density of $\{e_t\}$ and $A(z)=1-\sum_{k=1}^pa_kz^k$ is the AR polynomial. In particular,
\begin{equation}
\mathbb{E}[|e_t|^2] =\gamma_0^{(e)} = \frac{1}{2\pi}\int_{-\pi}^{\pi}f_e(\theta)d\theta=\sigma_p^2.
\end{equation}
Since $R^{(p)}\succ0$ and $\sigma_p^2>0$, $A(z)$ has no zeros for $|z|\leq1$. Jensen's formula then gives $\frac{1}{2\pi}\int_{-\pi}^{\pi}\log|A(e^{-i\theta})|^2d\theta=\log|A(0)|^2=0$, and therefore
\begin{align}
\frac{1}{2\pi}\int_{-\pi}^{\pi}\log f(\theta)d\theta
&= \frac{1}{2\pi}\int_{-\pi}^{\pi}\log f_e(\theta)d\theta \nonumber\\
&\leq \log\left[\frac{1}{2\pi}\int_{-\pi}^{\pi}f_e(\theta)d\theta\right]\nonumber\\
&= \log\sigma_p^2,
\end{align}
where the inequality is Jensen's inequality for the concave logarithm. Equality holds if and only if $f_e$ is constant almost everywhere, i.e. $f_e=\sigma_p^2$. Therefore, the spectral density of the process maximizing the differential entropy rate is
\begin{equation}
     f_{\rm ME}(\theta) = \frac{\sigma_p^2}{\left|1-\sum_{k=1}^{p}a_ke^{-ik\theta}\right|^2},
\end{equation}
which is the spectral density of an AR($p$) introduced in \cref{app:eq:spectral_arp} with covariance lags compatible with $\gamma_0,\dots,\gamma_p$. Therefore $f_{\rm ME}(\theta)$ is the unique solution of \cref{app:eq:burg_spectral_problem,app:eq:spectral_constraints}, and $\max \bar{h}(\Xt)=\log(\pi e \sigma^2_p)$, in agreement with what we found in the proof of \cref{app:thm:burg}.

\section{More about the differences with standard linear prediction}
\label{app:sec:linear_prediction}
Here, we highlight some important differences between our approach and standard LP~\cite{Barthel_2009_spectral, White_2008_spectral, Tang_2025}. In LP, the sampled correlator time series $\{g_k\}_{k=0}^{n-1}$ is modeled by the recursion $g_t = \sum_{k=1}^{p} \tilde{a}_k g_{t-k}$, and the coefficients $ \{\tilde{a}_k\}$ are obtained by least-squares minimization of the one-step prediction error. At first glance, this may seem equivalent to our approach. The differences, however, are subtle but fundamental.

First, LP treats $\{g_k\}$ as a signal realization rather than as the autocovariance sequence of a WSS process, so the AR ansatz has no statistical justification; its justification is entirely physical. A signal consisting of $p$ complex exponentials, $g_k = \sum_{m=1}^{p} W_m e^{-i\omega_m k}$, satisfies a linear recurrence of order $p$, formally identical to the AR($p$) recursion in \cref{app:eq:ARp} without the innovation term. As discussed for the case of a singular $R^{(n)}$ in \cref{sec:arp}, this representation is exact for any finite system, but may require a number of frequencies growing exponentially with system size. It is therefore practically useful only for systems whose spectrum is dominated by a few frequencies. Even then, the short-time dynamics of typical many-body correlation functions is governed by interference among many frequencies, with the dominant ones emerging only after a transient. In practice, one therefore often discards the first few points of the time series to improve the quality of the fit. 

Second, the order $p$ of this AR($p$) model is a hyperparameter whose choice is often heuristic since numerical or experimental time series are usually too short to apply principled information theoretic criteria, such as AIC or BIC~\cite{StoicaMoses2005}. Choosing $p$ too small results in overdamping, while choosing it too large leads to spurious, unstable poles. A common practice is to choose $n/3<p<n/2$~\cite{Barthel_2009_spectral, Stroeks_2022, Erpenbeck_2026}, so the model has far fewer parameters than data points, in contrast with our me(V)AR approach, where the maximal order $p = n-1$ can be used. A common issue occurring when applying LP to many-body correlation functions is that some of the poles of the fitted AR model lie outside the unit circle. These must be regularized, either by projecting them onto the unit circle or by discarding them (setting their amplitudes to zero). Finally, LP ensures neither that the extended sequence is pSD (or, equivalently, that its spectral function is non-negative) nor that it satisfies the total spectral weight sum rule. 

It is worth noting YW-type normal equations can also arise in LP within a least-squares approach. In that case, however, the covariance lags are estimated from the finite time series as time-averaged products of the samples~\cite{Makhoul_1975_linear, StoicaMoses2005, CoverThomas2006}), e.g., $\hat{\gamma}_k=\frac{1}{n}\sum_j g_j g_{j-k}^*$,
as commonly done in LP approaches~\cite{Barthel_2009_spectral, Lin_2023}, rather than being directly identified with the correlator samples, $\gamma_k=g_k$. Such sample covariance estimates are themselves subject to finite-sample effects: the normalization above guarantees a PSD Toeplitz matrix but yields, in general, a biased estimator of the covariance lags, whereas the unbiased normalization, with $n$ replaced by $n-k$, has variance growing with $k$ and does not guarantee PSDness~\cite{Burg1975}. More importantly for our purposes, the difference between the two constructions is not merely cosmetic: the time-averaged lags are quadratic functionals of the samples $g_k$, and therefore define a different covariance sequence and, consequently, a spectral density different from the maximum-entropy-rate solution of \cref{eq:maxent_spectrum}.

\section{PSD-Toeplitz denoising}
\label{app:sec:psd_denoising}
In this section we briefly review the denoising technique introduced in Ref.~\cite{Kemper_2024} and generalize it to the matrix correlator case.
\subsection{Scalar PSD-Toeplitz denoising}
Let us start from a single correlation function $g(t)$, sampled as $g_k = g(k\Delta t)$. Since $g(t)$ is a PSD function, the associated Gram matrix
\begin{equation}
	\label{app:eq:gt_Gram}
	R^{(n)} = \begin{bmatrix}
		g_0 & g_1^*  & \cdots & g_{n-1}^*\\
		g_1 & g_0 & \cdots & g_{n-2}^*\\
		\vdots & \vdots  & \ddots  & \vdots\\
		g_{n-1} & g_{n-2} & \cdots & g_0\\
	\end{bmatrix}
\end{equation}
must be a Hermitian, Toeplitz, and PSD matrix. In general, a Gram matrix constructed from noisy measurements of $g(t)$ will in general violate positive semi-definiteness and, if its entries are not all obtained from a single sequence $\{g_k\}$ (e.g., if $g_{-k}$ is measured independently rather than set to $g_k^*$), also the Hermitian and Toeplitz structure. These properties must therefore be restored. Moreover, the value $g_0 = g(0)$ can often be computed analytically, and provides an additional mitigation strategy~\cite{Kemper_2024}. The measured Gram matrix can then be denoised by finding a matrix in the intersection of the convex sets of PSD and Hermitian Toeplitz matrices, e.g., using alternating projections or Dykstra's algorithm \cite{Grigoriadis_1994, Higham_1988, Higham_2002, Bauschke_1996, Eberle_2003}. Dykstra's algorithm returns the Frobenius-nearest such matrix, whereas alternating projections converge to a point in the intersection that is, in general, not the nearest one. In practice, we found alternating projections to be robust and fast. Each iteration consists of the following three steps (omitting the superscript for simplicity):
\begin{enumerate}
	\item{Hermitianize: $R\leftarrow (R^\dagger+R)/2$.}
	\item{Regularized PSD projection: $R\leftarrow V\tilde{\Lambda}V^\dagger$, where $R=V \Lambda V^\dagger$ is the eigendecomposition of $R$, $\Lambda = \mathrm{diag}(\lambda_0,\dots,\lambda_{n-1})$, and $\tilde{\Lambda} = \mathrm{diag}(\tilde{\lambda}_0,\dots,\tilde{\lambda}_{n-1})$ with $\tilde{\lambda}_i = \max(\lambda_i, \delta)$, $\delta\geq 0$.}
	\item{Toeplitz projection: $R_{ij} \leftarrow \bar{g}_{i-j}$, with 
	\begin{equation}
		\label{app:eq:toeplitz_average}
		\bar{g}_k= \frac{1}{n-k}\sum_{i=0}^{n-k-1} R_{i+k,i},
	\end{equation}
	for $k\geq0$ and $\bar{g}_{-k}= \bar{g}^*_k$ (that is, each diagonal of $R$ is replaced by the average of its entries).}
\end{enumerate}
The three steps are repeated until the Frobenius-norm change between successive iterates falls below a chosen tolerance. Note that in the PSD projection step we introduced an additional regularization step with respect to the original algorithm discussed in Ref.~\cite{Kemper_2024}: instead of requiring $R \succeq 0$, we project $R$ onto the set of matrices with $R\succeq\delta I$. This is particularly useful in the presence of noisy data. A margin $\delta > 0$ keeps $R$ strictly positive definite and also makes positive semi-definiteness robust to perturbations. For a Hermitian noise matrix $E$, Weyl's inequality~\cite{HornJohnson1991} gives $\lambda_{\rm min}(R+E) \geq \lambda_{\rm min}(R) + \lambda_{\rm min}(E) \geq \delta - d(E)$, where $d(E) = \max(0, -\lambda_{\rm min}(E))$. Hence, $R+E \succeq 0$ whenever $d(E)\leq\delta$. This suggests setting $\delta$ at the scale of the noise-induced negative eigenvalues, which we estimate via a parametric bootstrap. Given the measured correlator $g_k$ and its error $\sigma_k$, we draw $B$ bootstrap samples $g^{(b)}_k$, $b=1,\dots,B$ from the corresponding noise model (which, for simplicity, we can assume to be Gaussian) and form the Hermitian Toeplitz Gram matrices $E^{(b)}$ of the fluctuations $g^{(b)}_k - g_k$. We then set $\delta = c\,\delta_{95}$, where $\delta_{95}$ is the 95th percentile of $d(E^{(b)})$ and $c$ is a scale factor. Since the Weyl bound is a worst case over all directions, while instability is driven only by the near-null directions of $R$, we expect $c<1$. We found that $c = 0.1$ gives good results for the example considered in this work. 

\subsection{Matrix PSD-Toeplitz denoising}
The algorithm discussed in the previous section generalizes straightforwardly to the case of a matrix correlator $G_{ij}(t)$, sampled as $G_{ij,k} = G_{ij}(k\Delta t)$. The associated Hermitian PSD Gram matrix 
\begin{equation}
	\label{app:eq:Gt_block_gram}
	\mathcal{R}^{(n)} = 
	\begin{bmatrix}
		G_0 & G^\dagger_1  & \cdots & G^\dagger_{n-1}\\
		G_1 & G_0  & \cdots & G^\dagger_{n-2}\\
		\vdots  & \vdots & \ddots  & \vdots\\
		G_{n-1} & G_{n-2}  & \cdots & G_0\\
	\end{bmatrix}
\end{equation}
has a block-Toeplitz structure, i.e., it is constant along block diagonals, while the blocks $G_k$ themselves are in general not Toeplitz (nor, for $k\neq 0$, Hermitian). The Hermitianization and PSD projection steps apply unchanged, while the Toeplitz projection step should be modified is modified to account for the block structure: each block $\mathcal{R}^{(n)}_{rs}$ is replaced by $\bar{G}_{r-s}$, with 
\begin{equation}
	\bar{G}_k= \frac{1}{n-k}\sum_{r=0}^{n-k-1} \mathcal{R}^{(n)}_{r+k,r}
\end{equation}
for $k\geq0$ and $\bar{G}_{-k}= \bar{G}^\dagger_k$. The regularization parameter $\delta$ can be estimated as in the scalar case, with $E^{(b)}$ the block Gram matrices of the bootstrap fluctuations.

\subsection{PSD-Toeplitz denoising, imperfect state preparation, and Trotterization}
In near-term applications to quantum computing, one usually cannot prepare a stationary state of a system exactly, and time evolution is performed via approximate algorithms such  as Trotterization. Here, we briefly show how the PSD-Toeplitz denoising can, in principle, partially mitigate the errors arising in such circumstances. Let us start from errors in state preparation. For simplicity, we consider the scalar correlator case and focus on the preparation of the ground state. We can then write the initial state as 
\begin{equation}
\ket{\psi} = \sqrt{1-\varepsilon^2}\ket{0} + \varepsilon \ket{\chi},
\end{equation} 
with $\ket{0}$ the target ground state and $\ket{\chi} = \sum_{l>0}c_l\ket{l}$, with $\sum_{l>0}|c_l|^2 = 1$, where $\ket{l}$ are the eigenstates of the system Hamiltonian with energy $E_l$. For simplicity, we assume a non-degenerate spectrum. Since, in general, $\ket{\psi}$ is not a stationary state, we have to compute the two-time correlation function $g_{ij} = g(t_i, t_j)$, 
\begin{equation}
	g_{ij} = \tr[\rho A^{\dagger}(t_i)A(t_j)],
\end{equation}
where we have introduced the density matrix
\begin{align}
	\rho &= \ket{\psi}\bra{\psi} \nonumber\\
    &= (1-\varepsilon^2)\ket{0}\bra{0} \nonumber\\
	& + \varepsilon\sqrt{1-\varepsilon^2}\sum_{l>0}[c_l \ket{l}\bra{0} + c_l^*\ket{0}\bra{l}]\nonumber\\
	& + \varepsilon^2 \sum_{m,l > 0} c_m c_l^* \ket{m}\bra{l}. 
\end{align}
Introducing the center and relative times $T_{ij} = (t_i + t_j)/2$ and $\tau_{ij} = t_i-t_j$, with $t_j = j \Delta t$,  we have
\begin{equation}
	g_{ij} = \tr[\rho(T_{ij})B(\tau_{ij})],
\end{equation}
with $B(\tau_{ij}) = A^\dagger(\tau_{ij}/2)A(-\tau_{ij}/2)$ and $\rho(T) = e^{-iHT}\rho\,e^{iHT}$. Since, in general, $[H, \rho] \neq 0$, $\rho(T)$ is time dependent:
\begin{equation}
	\rho(T_{ij}) = \sum_{m,l} \rho_{ml} e^{-i\omega_{ml}T_{ij}}\ket{m}\bra{l},
\end{equation}
with $\omega_{ml} = E_m - E_l$, $\rho_{00} = (1-\varepsilon^2)$, $\rho_{l0} = \rho_{0l}^* = \varepsilon\sqrt{1-\varepsilon^2}\,c_l$ for $l>0$, and $\rho_{ml} = \varepsilon^2 c_m c^*_l$ for $m,l>0$. In this case, the Gram matrix of \cref{app:eq:gt_Gram} is still Hermitian and PSD but it lacks the Toeplitz structure:
\begin{equation}
    \label{app:eq:two_time_gram_matrix}
	R^{(n)} = \begin{bmatrix}
		g_{00} & g_{01}  & \cdots\\
		g_{10} & g_{11} & \cdots \\
		\vdots & \vdots  & \ddots
	\end{bmatrix}.
\end{equation}
Note that this requires measuring $\mathcal{O}(n^2)$ two-time correlation functions $g_{ij}$, instead of $\mathcal{O}(n)$, as in the stationary case.
If we now apply the PSD-Toeplitz denoising protocol, the Toeplitz projection replaces the elements $g_{i+k,i}$ of the $k$-th diagonal with their average,
\begin{align}
	\bar{g}_k &= \frac{1}{M_k}\sum_{i=0}^{M_k-1} \tr[\rho(T_{i+k,i})B({\tau_k})] \nonumber \\
	&= \tr[\bar{\rho}_kB({\tau_k})],
\end{align}
where we defined $M_k = n-k$, used that $B(\tau_{i+k,i}) = B({\tau_k})$ with $\tau_k = k\Delta t$, and introduced 
\begin{equation}
\bar{\rho}_k = \frac{1}{M_k} \sum_{i=0}^{M_k-1} \rho(T_{i+k,i}).
\end{equation}
Since $T_{i+k,i} = (i + k/2)\Delta t$, the generic matrix element of $\bar{\rho}_k$ is given by
\begin{equation}
	(\bar{\rho}_k)_{ml} = \rho_{ml}D_k(\omega_{ml}),
\end{equation}
with 
\begin{equation}
	D_k(\omega) = e^{-i\omega(n-1)\Delta t/2}\frac{\sin(M_k\omega\Delta t/2)}{M_k\sin(\omega\Delta t/2)}.
\end{equation}
Since $D_k(0) = 1$, the diagonal elements of $\rho$ are unaffected by the average. However, the off-diagonal components $\rho_{m,l}$ with $m \neq l$ are suppressed by a factor
\begin{align}
	|D_{k}(\omega_{ml})| &= \left|\frac{\sin(M_k\omega_{ml}\Delta t/2)}{M_k\sin(\omega_{ml}\Delta t/2)}\right| \nonumber \\
	&\simeq \left|\mathrm{sinc}\left(\frac{M_k \omega_{ml}\Delta t}{2}\right)\right|,
\end{align}
where the last approximation holds for $\omega_{ml}\Delta t \ll 1$. Note that $|D_{k}(\omega_{ml})|\leq 1$ and $|D_{k}(\omega_{ml})|\ll 1$ for $|\omega_{ml}|M_k \Delta t\gg1$ provided $\omega_{ml}\Delta t\neq 2\pi q$ with $q\in\mathbb{Z}$ (i.e., in the absence of aliasing). Expanding the correlation function $g_{ij}$ in powers of $\varepsilon$,
\begin{align}
	g_{ij} &= \bra{\psi}A^{\dagger}(t_i)A(t_j)\ket{\psi}\nonumber\\
	&=g^{(0)}_{ij} + \varepsilon g^{(1)}_{ij} + \varepsilon^2 g^{(2)}_{ij} + \mathcal{O}(\varepsilon^3), 
\end{align}
we see that Toeplitz averaging suppresses the contribution $g^{(1)}_{ij}$, which arises entirely from the coherences $\rho_{l0}$, and the off-diagonal part of $g^{(2)}_{ij}$, by a $k$- and $\omega_{ml}$-dependent factor. For $g^{(1)}_{ij}$, the relevant frequencies satisfy $|\omega_{l0}| \geq \Delta$, with $\Delta = E_1 - E_0$ the spectral gap, so that suppression requires $M_k\Delta t \gg 1/\Delta$. The suppression is weakest at the largest lags: for $k \to n-1$, $M_k \to 1$ and no averaging takes place.
This becomes particularly clear in the limit $n\to\infty$ at fixed $k$ (so that $M_k \to \infty$), where
\begin{equation}
	\bar\rho_k \rightarrow (1-\varepsilon^2)\ket{0}\bra{0} + \varepsilon^2 \sum_{l>0} |c_l|^2 \ket{l}\bra{l}
\end{equation}
converges to the diagonal ensemble $\forall k$, and only the diagonal contributions of order $\epsilon^2$ survive: 
\begin{equation}
\bar g_k \rightarrow (1-\varepsilon^2)\bra{0}A^\dagger(\tau_k)A\ket{0} + \varepsilon^2 \sum_{l>0}|c_l|^2\bra{l}A^\dagger(\tau_k)A\ket{l}.
\end{equation}
The Toeplitz projection thus reduces the leading state-preparation error from $\mathcal{O}(\varepsilon)$ to $\mathcal{O}(\varepsilon^2)$. The residual $\mathcal{O}(\varepsilon^2)$ term is itself the correlation function of a stationary state, the diagonal ensemble, and therefore cannot be removed by PSD-Toeplitz denoising.

A similar argument can be made for the error introduced by Trotterization. Denoting the (time-independent) Hamiltonian of the system by $H$, its time evolution operator is $U(t) = e^{-iHt}$.  Let us introduce the single-step Trotter evolution operator $\tilde{U}(\delta t) $. Note that $ \tilde{U}(\delta t) \neq e^{-iH\delta t}$. For simplicity, we assume $\delta t = \Delta t$, with $\Delta t$ the sampling interval of the correlator $g(t_i, t_j)$ (more generally, it suffices that $\Delta t$ is an integer multiple of a fixed Trotter step). The Trotterized two-time correlation function is
\begin{equation}
	\tilde{g}_{ij} = \tr[\rho \tilde{A}_i^\dagger \tilde{A}_j],
\end{equation}
with $\rho=\ket{0}\bra{0}$ the exact ground state density matrix and $\tilde{A}_i = \tilde{U}^{-i}A\tilde{U}^i$. During the Toeplitz projection, each $\tilde{g}_{i+k,i}$ is replaced with the diagonal average of  \cref{app:eq:toeplitz_average}. In this case, 
\begin{equation}
	\tilde{g}_{i+k,i} = \tr[\rho_i \tilde{B}_k],
\end{equation}
with $\rho_i = \tilde{U}^i \rho \tilde{U}^{-i}$ and $\tilde{B}_k = \tilde{A}_k^{\dagger}A$. Note that for the exact dynamics $\tilde{U} = U$ and $\rho_i = \rho$, and therefore we retrieve the stationary result $\tilde{g}_{i+k,i} = g_{k}$. Let us introduce the eigenbasis of the Trotter evolution operator, $\{\ket{\tilde{l}}\}$, with $\tilde{U}\ket{\tilde{l}} = e^{-i\theta_l}\ket{\tilde{l}} $. Expanding the ground state of $H$, $\ket{0}$, in this basis, we have
\begin{equation}
	\rho = \ket{0}\bra{0} = \sum_{m,l} \tilde{\rho}_{ml}\ket{\tilde{m}}\bra{\tilde{l}}.
\end{equation}
When replacing each $\tilde{g}_{i+k,i}$ with its average according to \cref{app:eq:toeplitz_average}, we obtain
\begin{equation}
	\bar{\tilde{g}}_k = \tr[\bar{\rho}_k \tilde{B}_k],
\end{equation}
with 
\begin{equation}
	\bar{\rho}_k = \frac{1}{M_k}\sum_{i =0}^{M_k-1} \tilde{U}^i \rho \tilde{U}^{-i}.
\end{equation}
In the eigenbasis of the Trotter evolution operator, the matrix elements of $\bar{\rho}_k$ are
\begin{equation}
    (\bar{\rho}_k)_{ml} = \tilde{\rho}_{ml}\frac{1}{M_k}\sum_{j=0}^{M_k-1}e^{-ij\theta_{ml}} = \tilde{\rho}_{ml}\tilde{D}_k(\theta _{ml}),
\end{equation}
with $\theta_{ml} = \theta_m - \theta_l$ and, as in the case of approximate ground-state preparation,
\begin{equation}
    \tilde{D}_k(\theta)= e^{-i(M_k-1)\theta/2}\frac{\sin(M_k \theta/2)}{M_k \sin(\theta/2)}.
\end{equation}

Hence, the off-diagonal elements of $\bar{\rho}_k$ are suppressed by a $k$- and $\theta_{ml}$-dependent factor $|\tilde{D}_k(\theta_{ml})|$, which is small for $|\theta_{ml}|M_k\gg 1$ provided $\theta_{ml}\neq 2\pi q$ with $q\in\mathbb{Z}$ (i.e., no aliasing). Note, however, that the unsuppressed diagonal contributions do not correspond to those of the target dynamics and, in general, the errors on the correlators will remain $\mathcal{O}(\delta t^p)$. Indeed, for a $p$-th order Trotter formula, we can write $\tilde{U}(\delta t) = e^{-iH_{\rm eff}\delta t}$ , with 
\begin{equation}
	H_{\rm eff} = H + \delta t^{p}K + \mathcal{O}(\delta t^{p+1}),
\end{equation}
and $H_{\rm eff}\ket{\tilde{l}} = \tilde{E}_l\ket{\tilde{l}}$, so that $\theta_l = \tilde{E}_l\,\delta t$ (mod $2\pi$). From perturbation theory, we find that $\ket{\tilde{0}} = \ket{0} + \mathcal{O}(\delta t^p)$, or, conversely,
\begin{equation}
\ket{0} = \sqrt{1-\varepsilon^2}\ket{\tilde{0}} + \varepsilon\sum_{l>0} b_l \ket{\tilde{l}}, \qquad \varepsilon = \mathcal{O}(\delta t^p),
\end{equation}
with $\sum_{l>0}|b_l|^2 = 1$. So, effectively, preparing the system in $\ket{0}$ and evolving it with $H_{\rm eff}$ falls into the previous case of approximate ground-state preparation, with target state $\ket{\tilde 0}$ and error $\varepsilon = \mathcal{O}(\delta t^p)$. In the limit $n \to \infty$ at fixed $k$, the Toeplitz average therefore yields
\begin{equation}
\bar{\tilde{g}}_k \rightarrow \bra{\tilde 0}\tilde{A}_k^\dagger A\ket{\tilde 0} + \mathcal{O}(\delta t^{2p}),
\end{equation}
i.e., the correlation function of the Trotterized dynamics in its own ground state. Its frequencies and spectral weights both differ from the exact ones at order $\delta t^p$:
\begin{align}
\tilde{E}_l - \tilde{E}_0 &= E_l - E_0 + \mathcal{O}(\delta t^p), \\
|\bra{\tilde l}A\ket{\tilde 0}|^2 &= |\bra{l}A\ket{0}|^2 + \mathcal{O}(\delta t^p),
\end{align}
where the quasi-energy differences are defined modulo $2\pi/\delta t$. So, the Toeplitz average can only suppress the $ \mathcal{O}(\delta t^p)$ contribution due to the ground state mismatch, but not the one due to the different effective dynamics. Since $\ket{\tilde 0}$ is an eigenstate of $\tilde U$, the latter is the correlation function of an exactly stationary process, and hence a valid PSD Toeplitz sequence: the Trotter error is a bias that PSD-Toeplitz denoising cannot detect.

As noted below \cref{app:eq:two_time_gram_matrix}, mitigating ground state preparation and Trotterization errors via the PSD-Toeplitz denoising protocol described above requires measuring the two-time correlation function $g(t_i, t_j)$ and, therefore, requires $O(n^2)$ measurements instead of $O(n)$, as in the stationary case.

\end{document}